\documentclass[12pt,a4paper]{article}
\usepackage{amsmath}
\usepackage{amssymb}
\usepackage{amsthm}
\usepackage{amsfonts}
\usepackage[all]{xy}
\usepackage{cite}
\usepackage{hyperref}
\usepackage{color}
\usepackage{tikz}

\newtheorem{theorem}{Theorem}
\newtheorem{proposition}[theorem]{Proposition}

\newtheorem{lemma}[theorem]{Lemma}
\newtheorem{corollary}[theorem]{Corollary}

\newtheorem{claim}{Claim}

\newtheorem{observation}[theorem]{Observation}

\def\tobedone{\textcolor{red}{TO BE DONE.}}

\def\diam{\operatorname{diam}}

\renewenvironment{proof}{
\par
\noindent {\bf Proof.}\rm}{\mbox{}\hfill$\square$\par\vskip 3mm}
\newcommand\EFFACE[1]{}

\def\diam{{\rm diam}}
\def\rad{{\rm rad}}

\newcommand\LIGNE[4]{
\draw[thick] (#1,#2) to (#3,#4);
}
\newcommand\POINTILLE[4]{
\draw[thick,dotted] (#1,#2) to (#3,#4);
}

\newcommand\TIRET[4]{
\draw[thick,dashed] (#1,#2) to (#3,#4);
}

\newcommand\bSOM[4]{
   \node[scale=0.7,draw,circle,fill=black] at (#1,#2){};
   \node[above] at (#1,#2+0.2){#3};
   \node[below] at (#1,#2-0.2){#4};
}
\newcommand\gSOM[4]{
   \node[scale=0.7,draw,circle,fill=lightgray] at (#1,#2){};
   \node[above] at (#1,#2+0.2){#3};
   \node[below] at (#1,#2-0.2){#4};
}
\newcommand\SOM[4]{
   \node[scale=0.7,draw,circle,fill=white] at (#1,#2){};
   \node[above] at (#1,#2+0.2){#3};
   \node[below] at (#1,#2-0.2){#4};
}
\newcommand\drSOM[4]{
    \node[scale=0.7,draw,circle,fill=white] at (#1,#2){};
   \node[above] at (#1-0.4,#2){#3};
   \node[below] at (#1+0.4,#2){#4};
}
\newcommand\ghSOM[4]{
   \node[scale=0.7,draw,circle,fill=white] at (#1,#2){};
   \node[above] at (#1+0.4,#2){#3};
   \node[below] at (#1-0.9,#2+0.4){#4};

}
\newcommand\bdrSOM[4]{
    \node[scale=0.7,draw,circle,fill=black] at (#1,#2){};
   \node[above] at (#1-0.4,#2){#3};
   \node[below] at (#1+0.4,#2){#4};
}

\begin{document}

\title{On the Broadcast Independence Number\\ of Circulant Graphs}

\author{Abdelamin LAOUAR~\thanks{Faculty of Mathematics, Laboratory L'IFORCE, University of Sciences and Technology
Houari Boumediene (USTHB), B.P.~32 El-Alia, Bab-Ezzouar, 16111 Algiers, Algeria.}
\and Isma BOUCHEMAKH~\footnotemark[1]
\and \'Eric SOPENA~\thanks{Univ. Bordeaux, CNRS, Bordeaux INP, LaBRI, UMR5800, F-33400 Talence, France.}
}
\maketitle


\begin{abstract}
An independent broadcast on a graph $G$ is a function  $f: V
\longrightarrow \{0,\ldots,\diam(G)\}$ such that
$(i)$ $f(v)\leq e(v)$ for every vertex $v\in V(G)$, where
$\operatorname{diam}(G)$ denotes the diameter of $G$ and $e(v)$ the eccentricity of vertex $v$,
and $(ii)$ $d(u,v) > \max \{f(u), f(v)\}$ for every two distinct vertices $u$ and $v$ with $f(u)f(v)>0$.
The broadcast independence number $\beta_b(G)$ of $G$ is then the maximum value of $\sum_{v \in V} f(v)$, taken over all independent broadcasts on~$G$.

We prove that every circulant graph of the form $C(n;1,a)$,
$3\le a\le \lfloor\frac{n}{2} \rfloor$,
admits an optimal
$2$-bounded independent  broadcast, that is, an independent broadcast~$f$ satisfying
$f(v)\le 2$ for every vertex $v$,
except when $n=2a+1$, or $n=2a$ and $a$ is even.
We then determine the broadcast independence number of various
classes of such circulant graphs, and prove that, for most of these classes,
the equality $\beta_b(C(n;1,a)) = \alpha(C(n;1,a))$ holds,
where $\alpha(C(n;1,a))$ denotes the independence number of $C(n;1,a)$.


\end{abstract}


\medskip


\section{Introduction}\label{sec:intro}

All the graphs considered in this paper are undirected and simple.
For such a graph $G$, we denote by $V(G)$ and $E(G)$ its set of vertices and its set of
edges, respectively.
Let $G$ be a nontrivial connected graph, that is, a connected graph with at least one edge.
The \emph{distance} from a vertex $u$ to a vertex $v$ in $G$, denoted $d_G(u,v)$, or simply $d(u,v)$ when
$G$ is clear from the context, is the length (number of edges) of a shortest path from $u$ to $v$.
The \emph{eccentricity} of a vertex $v$ in $G$,
denoted $e_G(v)$, is the maximum distance from $v$ to any other vertex of $G$.
The minimum eccentricity in $G$ is the \emph{radius} of $G$, denoted $\rad(G)$, while the maximum eccentricity in $G$ is its \emph{diameter}, denoted $\diam(G)$.
Two vertices $u$ and $v$ with $d_G(u,v)=\diam(G)$ are said to be \emph{antipodal}.

A function $f: V(G) \rightarrow \{0,\ldots,\diam(G)\}$ is a \emph{broadcast on $G$} if
$f(v) \leq e_G(v)$ for every vertex $v \in V$.
For each vertex $v$, $f(v)$ is the \emph{$f$-value} of $v$, or the \emph{broadcast value} of $v$ if $f$ is clear from the context.
Given such a broadcast $f$, an \emph{$f$-broadcast vertex} 
is a vertex $v$ for which $f(v) > 0$.
The set of all $f$-broadcast vertices is denoted $V^+_f(G)$.
If $v$ is a broadcast vertex and $u$ a vertex such that $d(u,v) \leq f(u)$, then
the vertex $v$ \emph{$f$-dominates} the vertex $u$.
The \emph{cost} of a broadcast $f$ on $G$ is the value $\sigma(f) =\sum_{v\in V_f^+} f(v)$.

A broadcast $f$ is \emph{independent} if no broadcast vertex $f$-dominates
another broadcast vertex, or, equivalently, if $d(u,v) > \max\{f(u), f(v)\}$ for every two distinct broadcast vertices $u$ and $v$.
The maximum cost of an independent broadcast on $G$ is the \emph{broadcast independence number}
of $G$, denoted $\beta_b(G)$.
An independent broadcast with cost $\beta_b(G)$ is referred to as a \emph{$\beta_b$-broadcast}.

A subset $S$ of $V(G)$ is an \emph{independent set} if no two vertices in $S$ are adjacent in $G$.
The \emph{independence number} of $G$, denoted $\alpha(G)$, is then the maximum cardinality of an independent set in $G$.
Note that the characteristic function $f_S$ of every maximal independent set $S$ in a graph $G$
is an independent broadcast and, therefore, $\alpha(G)\leq \beta_b(G)$ for every graph $G$.

Broadcast independence and broadcast domination were introduced by Erwin~\cite{Erw01} in his Ph.D. dissertation, using the terms cost independence and cost domination, respectively.
He also discussed several other types of broadcast parameters and gave relationships between them.
Most of the corresponding results are published in~\cite{DEHHH,Erw04}.
Since then, several papers have been devoted to the study of these broadcast parameters,
but there were not so many results concerning the broadcast independence number~\cite{BZ12,DEHHH}
until recently (see \cite{ABS18,ABS19,BBS19,BR1_2018,BR2_2018,BR3_2018}).
In particular, Bessy and Rautenbach discussed the algorithmic complexity of broadcast independence
in~\cite{BR1_2018} and
the links between girth, minimum degree, independence number and broadcast independence number
in~\cite{BR2_2018,BR3_2018}.

\medskip

In this paper, we study the broadcast independence number of circulant graphs.
Recall that for every integer $n\ge 3$, and every sequence of integers $a_1,\dots,a_k$, $k\ge 1$, satisfying
$1 \leq a_1 \leq \dots \leq a_k\leq \left\lfloor \frac{n}{2}\right\rfloor$,
the \emph{circulant graph} $G = C(n;a_1,\dots,a_k)$ is
the graph defined by
$$V(G) = \{v_0,v_1,\dots,v_{n-1}\}\ \mbox{and }
E(G) = \left\{v_iv_{i+a_j}\ |\ a_j\in \{a_1,\dots,a_k\} \right\}$$
(subscripts are taken modulo $n$).
Note that, in particular, $C(n;a_1,\dots,a_k)$ is $2k$-regular and vertex-transitive
(see~\cite{M12} for a survey on properties of undirected circulant graphs).

\medskip


Our paper is organized as follows.
In Section~\ref{sec:preliminaries}, we give some preliminary results and determine the broadcast independence number of circulant graphs of the form $C(2a;1,a)$ and $C(3a;1,a)$.
In Section~\ref{sec:2-bounded} we prove that almost all circulant graphs of the form $C(n;1,a)$ admit an optimal independent broadcast all whose broadcast values are at most 2.
General upper and lower bounds on the cost of independent broadcasts on circulant graphs of the form $C(n;1,a)$ are proposed in Section~\ref{sec:bounds}.
We then determine the value of the broadcast independence number of various classes of circulant graphs in Section~\ref{sec:exact}.
We finally propose a few concluding remarks in Section~\ref{sec:discussion}.

\section{Preliminary results}\label{sec:preliminaries}


Let $\mu(G)$ denote the maximum cardinality of a set of pairwise antipodal vertices in $G$.
Dunbar {\it et al.} proved the following lower bound on the broadcast independence number
of a graph.

\begin{proposition}[Dunbar {\it et al.}~\cite{DEHHH}]\label{prop:lower bound}
For every graph $G$,
$$\beta_b(G)\geq \mu(G)(\diam(G)-1) \geq 2(\diam(G)-1).$$
Moreover, this bound is sharp.
\end{proposition}

In addition to grid graphs $G_{m,n}=P_m\,\Box\,P_n$ with $m \in \{2,3,4\}$ and $m\leq n$~\cite{BZ12} and
paths~\cite{Erw01}, the relation $\beta_b(G)= 2(\diam(G)-1)$ also holds for cycles
of order at least~4~\cite{BBS19}.

It can also be observed that the value $2(\diam(G)-1)$ is an upper bound on the cost of some independent broadcasts.

In order to compare the values of the independence number
and of the broadcast independence number  of the graphs
we will consider in Section~\ref{sec:exact}, the following observation will be useful.

\begin{observation}\label{obs:beta=alpha}
For every graph $G$, $\beta_b(G)\ge\alpha(G)$.
Moreover, $\beta_b(G)=\alpha(G)$ if and only if there exists
a $\beta_b$-broadcast $f$ on $G$ such that $f(v)=1$ for every broadcast
vertex $v\in V_f^+$.
\end{observation}


Indeed, the fact that the characteristic function $f_S$
of every maximal independent set $S$ in a graph $G$ is an independent broadcast on $G$,
as noticed in the previous section, gives the inequality and the necessity of the condition for the second part of the statement, while
the sufficiency follows from the fact that $V_f^+$ is always an independent set.


Before considering general cases in the next sections, we will determine in the following the independence number and the broadcast independence number of circulant graphs of the form $C(n;1,a)$ for two particular cases, namely when $n=2a$ or $n=3a$.

\begin{lemma}\label{lem:alphaC(2a;1,a)}
For every integer $a\ge 2$,
$$\alpha(C(2a;1,a)) =
\left\{ \begin{array}{ll}
   a  , & \text{if $a$ is odd,}\\[1ex]
   a - 1 ,       & \text{if $a$ is even.}
   \end{array}
\right.
$$
\end{lemma}

\begin{proof}
Since $C(4;1,2)=K_4$ and  $\alpha(K_4)=1$, we can assume $a\ge 3$.

If $a$ is odd, then  the set
$S = \{v_i|\ \text{$i$ is even} \}$ is an independent set of $C(2a;1,a)$, and thus
$\alpha(C(2a;1,a)) \ge |S| = a$.
Since $C_{2a}$ is a subgraph of $C(2a;1,a)$, we get
$\alpha(C(2a;1,a)\leq \alpha(C_{2a})=a$ and the result follows.

If $a$ is even, then the set
$S' =\{v_i|\ 0\le i \le a-2,\ \text{$i$ is even}\} \cup \{v_i|\ a+1 \le i \le 2a-3,\ \text{$i$ is odd}\}$ is an independent set of $C(2a;1,a)$, and thus
$\alpha(C(2a;1,a)) \ge |S'| = a-1$.
Note that the odd cycle $C=v_0v_1\dots v_av_0$, with $\alpha(C)=\frac{a}{2}$, is a subgraph of $C(2a;1,a)$.
Therefore, for every independent set $I$ of $C(2a;1,a)$,
there are at least two consecutive vertices $v_i$, $v_{i+1}$ with $v_i,v_{i+1}\notin I$.
This implies $\alpha(C(2a;1,a))\leq \frac{2a-1}{2}$,
which gives $\alpha(C(2a;1,a))=a-1$.
\end{proof}


\begin{figure}
\begin{center}
\begin{tikzpicture}[scale=0.7]

  \LIGNE {1}{0}{7.5}{0}
  \POINTILLE {7.5}{0}{8.5}{0}

  \POINTILLE {9.5}{0}{10.5}{0}
  \LIGNE {10.5}{0}{13.5}{0}
  \POINTILLE {13.5}{0}{14.5}{0}

  \POINTILLE {15.5}{0}{16.5}{0}
  \LIGNE {16.5}{0}{21}{0}

  \LIGNE {1}{-2}{7.5}{-2}
  \POINTILLE {7.5}{-2}{8.5}{-2}

  \POINTILLE {9.5}{-2}{10.5}{-2}
  \LIGNE {10.5}{-2}{13.5}{-2}
  \POINTILLE {13.5}{-2}{14.5}{-2}

  \POINTILLE {15.5}{-2}{16.5}{-2}
  \LIGNE {16.5}{-2}{21}{-2}

\TIRET {1}{0}{21}{-2}
\TIRET {1}{-2}{21}{0}

\LIGNE {1}{0}{1}{-2}
\LIGNE {3}{0}{3}{-2}
\LIGNE {5}{0}{5}{-2}
\LIGNE {7}{0}{7}{-2}
\LIGNE {11}{0}{11}{-2}
\LIGNE {13}{0}{13}{-2}
\LIGNE {17}{0}{17}{-2}
\LIGNE {19}{0}{19}{-2}
\LIGNE {21}{0}{21}{-2}

 \SOM{1}{0}{$v_{0}$}{}
  \SOM{3}{0}{$v_{1}$}{}
  \SOM{5}{0}{$v_{2}$}{}
  \SOM{7}{0}{$v_{3}$}{}
  \SOM{11}{0}{$v_{k}$}{}
  \SOM{13}{0}{$v_{k+1}$}{}
  \SOM{17}{0}{$v_{a-3}$}{}
  \SOM{19}{0}{$v_{a-2}$}{}
  \SOM{21}{0}{$v_{a-1}$}{}

\SOM{1}{-2}{}{$v_{a}$}
  \SOM{3}{-2}{}{$v_{a+1}$}
  \SOM{5}{-2}{}{$v_{a+2}$}
  \SOM{7}{-2}{}{$v_{a+3}$}
  \SOM{11}{-2}{}{$v_{a+k}$}
  \SOM{13}{-2}{}{$v_{a+k+1}$}
  \SOM{17}{-2}{}{$v_{2a-3}$}
  \SOM{19}{-2}{}{$v_{2a-2}$}
  \SOM{21}{-2}{}{$v_{2a-1}$}

\end{tikzpicture}
\caption{\label{fig:C(2a;1,a)} The circulant graph C(2a;1,a).}
\end{center}
\end{figure}
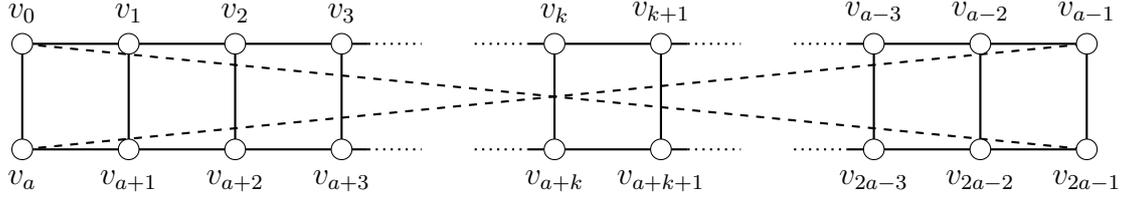

\begin{theorem}\label{th:C(2a;1,a)}
For every integer $a\ge 2$,
$$\beta_b(C(2a;1,a)) =
\left\{ \begin{array}{ll}
  \alpha(C(2a;1,a)) = a , & \text{if a is odd},\\ [1ex]
  \alpha(C(2a;1,a)) = a - 1 , & \text{if $a = 2^p$ for some integer $p\ge 1$},\\ [1ex]
   a        & \text{otherwise.}
   \end{array}
\right.
$$
\end{theorem}

\begin{proof}
The case $a=2$ directly follows from Lemma~\ref{lem:alphaC(2a;1,a)}.
We can thus assume $a\ge 3$.
The graph $C(2a;1,a)$ can be viewed as the Cartesian product graph $P_a\,\Box\, K_2$ with two additional edges (see Figure~\ref{fig:C(2a;1,a)}, where the two additional edges are drawn as dashed lines).
Recall that $\alpha(C(2a;1,a))$ is given by Lemma~\ref{lem:alphaC(2a;1,a)}, and
let $f$ be any independent $\beta_b$-broadcast on $C(2a;1,a)$.
If $\vert V_f^+ \vert = 1$, then
$$\sigma(f) = \diam(C(2a;1,a)) = \left\lfloor\frac{a+1}{2}\right\rfloor < \alpha(C(2a;1,a)),$$
which gives  $\sigma(f) < \beta_b(C(2a;1,a))$ by Observation~\ref{obs:beta=alpha}, a contradiction.
Therefore, $\vert V_f^+ \vert \geq 2$.
Since each vertex $v\in V_f^+$ $f$-dominates exactly $4f(v)$ vertices,
each $f$-broadcast vertex $v$ is $f$-dominated exactly once,
each non-broadcast vertex is $f$-dominated at most three times,
and at most $\vert V_f^+\vert$ vertices can be dominated three times (namely the vertices $v_{i+a}$ when $v_i\in V_f^+$), we get
\begin{equation}\label{eq:1}
4f(V_f^+)  \leq 3|V_f^+|+|V_f^+| + 2\left(2a - 2|V_f^+|\right)= 4a,
\end{equation}
and thus
\begin{equation}\label{eq:2}
\beta_b(C(2a;1,a)) = \sigma(f) = \sum_{v\in V_f^+}f(v) = f(V_f^+) \leq a.
\end{equation}

We now consider the three cases of the statement of the theorem separately.

\begin{enumerate}
\item   $a$ is odd.\\
Let $g$ be the mapping from $V(C(2a;1,a))$ to $\{0,1\}$
defined by  $g(v_i)=1$ if and only if $i$ is even.
Since $a$ is odd, $g$ is an independent broadcast on $C(2a;1,a)$.
This gives $\beta_b(C(2a;1,a)) \ge \sigma(g) =  a$ and, since $g$ satisfies~(\ref{eq:2}), $\beta_b(C(2a;1,a))= a$.
By Observation~\ref{obs:beta=alpha}, we then get
$\beta_b(C(2a;1,a)) = \alpha(C(2a;1,a)) = a.$

\item  $a = 2^p$ for some integer $p\ge 1$.\\
By Observation~\ref{obs:beta=alpha} and Lemma~\ref{lem:alphaC(2a;1,a)}, we have
$\beta_b(C(2a;1,a)) \geq \alpha(C(2a;1,a)) = a-1$, since $a$ is even.
Let $g$ be any independent $\beta_b$-broadcast on $C(2a;1,a)$.

Suppose first that not all $g$-broadcast vertices have the same $g$-value,
and let $v_i$ and $v_j$ be any two vertices with $g(v_i) < g(v_j)$ such that the distance $d(v_i,v_j)$ is minimum among all $g$-broadcast vertices with distinct $g$-values. Without loss of generality, we can assume $i<j$. We consider two
subcases, depending on whether $v_i$ and $v_j$ are on the same side of the ``ladder'' (refer to Figure~\ref{fig:C(2a;1,a)}) or not.
\begin{enumerate}
\item {\it $j-i>a$ ($v_i$ and $v_j$ are not on the same side of the ladder).}\\
Since no $g$-broadcast vertex lies on a shortest path linking $v_i$ and $v_j$, $v_{j+a}$ is not $g$-dominated by $v_i$ and is thus $g$-dominated at most twice.
Therefore, the inequality~(\ref{eq:1}) becomes
$$4g(V_g^+)  \leq  3(|V_g^+|-1)+|V_g^+| + 2\left(2a - 2|V_g^+|+1\right) = 4a-1,$$
which gives
$$\beta_b(C(2a;1,a)) = \sigma(g) = \sum_{v\in V_g^+}g(v) = g(V_g^+) \leq \left\lfloor\frac{4a-1}{4}\right\rfloor=a-1.$$

\item {\it $j-i<a$ ($v_i$ and $v_j$ are on the same side of the ladder).}\\
If there exists a $g$-broadcast vertex $v_k$ with $i+a<k<j+a$ then,
since no $g$-broadcast vertex lies on a shortest path linking $v_i$ and $v_j$,
we necessarily have $v_k\in\{v_{i+a+1},v_{j+a-1}\}$. By considering either $v_i$ and $v_k$, or $v_j$ and $v_k$, instead of $v_i$ and $v_j$, we are back to the previous subcase.

If no such vertex exists, then both $v_{i+a}$ and $v_{j+a}$ are $g$-dominated at most twice and thus, using the same argument as before, we get $\beta_b(C(2a;1,a))\le a-1$.
\end{enumerate}

We thus get $\beta_b(C(2a;1,a)) = a-1$ in both subcases, as required.
Suppose now that $g(v)=k$ for every vertex $v\in V_g^+$ and let $v_i$ be any such vertex.

If $v_{i+a+k}\notin V_g^+$, then  the vertex $v_{i+a}$ is $g$-dominated at most twice and thus, as previously, we get $\beta_b(C(2a;1,a)) = \sigma(g) = a-1$.
The same conclusion arises if $v_{i-a-k}\notin V_g^+$.

Suppose finally
$v_{i+a+k},v_{i-a-k}\in V_g^+$
for every vertex $v_i \in V_g^+$ and assume, without loss of generality, $v_0\in V_g^+$.
We thus have (recall that indices are taken modulo $2a$)
\begin{align*}
V_g^+ & = \{v_{0},v_{a+k},v_{2k},v_{a+3k}, \ldots\},\\
      & = \{v_{0},v_{2k},v_{4k}, \ldots ,v_{a+k},v_{a+3k}, v_{a+5k}, \ldots \}.
\end{align*}
Hence, every $g$-broadcast vertex $v_i$ with $0\le i<a$ satisfies $i=2kt$ for some $t$, $0\le t< \frac{a}{2k}$.
Since $v_{-a-k}=v_{a-k}$ is a $g$-broadcast vertex and $0<a-k<a$, we get
$a-k=2kt'$, for some $t'$, $0< t'< \frac{a}{2k}$.
This gives $a=(2t'+1)k$, contradicting the assumption $a=2^p$, so that this last case cannot appear.

\item $a$ is even and $a \neq 2^p$ for every $p>0$.\\
This implies  $a=(2\ell+1)2^k$ for some positive integers $k$ and $\ell$.
Let $g$ be the mapping from $V(C(2a;1,a))$ to $\{0,2^k\}$
defined by  $g(v_i)=2^k$ if and only if $i\equiv 0 \pmod{2^{k+1}}$,
which gives $V_g^+=\{v_{p2^{k+1}}|\ 0\le p\le 2\ell\}$.

For any two $g$-broadcast vertices $v_i=v_{p2^{k+1}}$ and $v_j=v_{q2^{k+1}}$, $0\le p < q \le 2\ell$, we have
$$d(v_i,v_j) = \min\{ (q-p)2^{k+1}, |(p-q)2^{k+1}+a| + 1, (p-q)2^{k+1} + 2a\},$$
which gives, since $a=(2\ell+1)2^k$,
$$d(v_i,v_j) = \min\{ (q-p)2^{k+1}, |2(p-q+\ell)+1|2^k + 1, (p-q+2\ell+1)2^{k+1}\}\ge 2^k+1.$$
Therefore, $g$ is an independent broadcast on $C(2a;1,a)$, with cost
$\sigma(g) = 2^k \left(\frac{2a}{2^{k+1}}\right)=a$, which gives
$\beta_b(C(2a;1,a)) \ge \sigma(g)= a$ and thus, since $g$ satisfies (\ref{eq:2}), $\beta_b(C(2a;1,a)) = a$.

%
%
%

\end{enumerate}

This completes the proof.
\end{proof}

We finally consider the case $n=3a$.

\begin{theorem}\label{th:C(3a;1,a)}
For every integer $a\ge 3$,
$$\beta_b(C(3a;1,a))=\alpha(C(3a;1,a))=a.$$
\end{theorem}

\begin{proof}
Let $f$ be an independent $\beta_b$-broadcast on $C(3a;1,a))$.
For each vertex $v_i\in V_f^+$, we let
$$C_f^i =  \{v_i,\dots,v_{i+f(v_i)-1}\}
       \ \cup\ \{v_{i+a},\dots,v_{i+a+f(v_i)-1}\}
       \ \cup\  \{v_{i+2a},\dots,v_{i+2a+f(v_i)-1}\}.$$
We clearly have $|C_f^i|=3f(v_i)$ for every $v_i\in V_f^+$,
and $C_f^i \cap C_f^{i'} = \emptyset$ for every two distinct vertices $v_i$ and $v_{i'}$ in $V_f^+$,
since otherwise we would have $d(v_{i},v_{i'})\leq \max \{f(v_{i}),f(v_{i'})\}$,
contradicting the fact that $f$ is an independent broadcast.
This gives
$$3f(V_f^+) = \sum_{v_i\in V_f^+}|C_f^i| \leq 3a,$$
and thus
$$\beta_b(C(3a;1,a)) = \sigma(f) = \sum_{v\in V_f^+}f(v) = f(V_f^+) \leq \frac{3a}{3}=a.$$

Consider now the mapping $f$ from $V(C(3a;1,a))$ to $\{0,1\}$ defined as follows,
depending on the parity of $a$.
\begin{enumerate}
\item If $a$ is odd, then $f(v_i)=1$ if and only if $i$ is even and $i < 2a$.

\item If $a$ is even, then $f(v_i)=1$ if and only if $(i \mod a+1)$ is odd and $i \leq 2a$.
\end{enumerate}

In both cases, $f$ is clearly an independent broadcast on $V(C(3a;1,a)$
with $\sigma(f)=a$
%
This implies $\beta_b(C(3a;1,a)) \ge a$ and thus,
thanks to Observation~\ref{obs:beta=alpha},
$\beta_b(C(3a;1,a)) = \alpha(C(3a;1,a)) = a$.
\end{proof}


\section{2-bounded optimal independent broadcasts}\label{sec:2-bounded}

Recall that we denote by $v_0,v_1,\dots,v_{n-1}$ the vertices
of $C(n;1,a)$ and that subscripts are always considered modulo $n$.
We will say that an edge $v_iv_j$ is a \emph{$k$-edge} for some integer $k$,
$1\le k\le\left\lfloor\frac{n}{2}\right\rfloor$, if $j=i+k$ or $i=j+k$.
Therefore, every edge in $C(n;1,a)$ is either a $1$-edge or an $a$-edge.

Let $f$ be an independent broadcast on $C(n;1,a)$.
For every $f$-broadcast vertex $v_i\in V_f^+$, we denote by $D_f(v_i)$
the set of vertices that are $f$-dominated by $v_i$, that is
$$\begin{array}{rcl}
D_f(v_i) & = & \bigcup_{k=0}^{f(v_{i})} \left\{v_j\ |\ i-(f(v_{i})-k)a-k \le j\le i-(f(v_{i})-k)a+k  \right\} \\ [2ex]
 & \cup & \bigcup_{k=0}^{f(v_{i})} \left\{v_j\ |\  i+(f(v_{i})-k)a-k \le j\le i+(f(v_{i})-k)a+k  \right\}.
\end{array}
$$


\begin{figure}
\begin{center}
\begin{tikzpicture}[scale=0.5]
   \POINTILLE{-1}{0}{0}{0}
   \LIGNE{0}{0}{26}{0}
   \POINTILLE{26}{0}{27}{0}

   \foreach \k in {1,7,13,19}
      \draw[thick] (\k,0) .. controls (\k+1,2) and (\k+5,2) .. (\k+6,0);

   \SOM{0}{0}{}{}
   \gSOM{1}{0}{}{}
   \SOM{2}{0}{}{}
   \SOM{3}{0}{}{}
   \SOM{4}{0}{}{}
   \SOM{5}{0}{}{}
   \gSOM{6}{0}{}{}
   \gSOM{7}{0}{}{}
   \gSOM{8}{0}{}{}
   \SOM{9}{0}{}{}
   \SOM{10}{0}{}{}
   \gSOM{11}{0}{}{}
   \gSOM{12}{0}{}{}
         \bSOM{13}{0}{2}{$v_i$}
   \SOM{26}{0}{}{}
   \gSOM{25}{0}{}{}
   \SOM{24}{0}{}{}
   \SOM{23}{0}{}{}
   \SOM{22}{0}{}{}
   \SOM{21}{0}{}{}
   \gSOM{20}{0}{}{}
   \gSOM{19}{0}{}{}
   \gSOM{18}{0}{}{}
   \SOM{17}{0}{}{}
   \SOM{16}{0}{}{}
   \gSOM{15}{0}{}{}
   \gSOM{14}{0}{}{}
\tobedone
\end{tikzpicture}
\caption{\label{fig:Df(vi)}The set $D_f(v_i)$ (black vertex and grey vertices), with $a=6$ and $f(v_i)=2$.}
\end{center}
\end{figure}
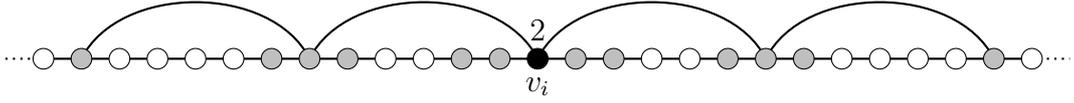

Figure~\ref{fig:Df(vi)} illustrates this definition on a circulant
graph of the form $C(n;1,6)$ (with $n\ge 26$) for a vertex $v_i$ with $f(v_i)=2$.

\medskip


Let us say that an independent broadcast $f$ is \emph{$\ell$-bounded}, for some integer $\ell\ge 1$, if $f(v)\le \ell$ for every vertex $v$. In particular, a $1$-bounded independent broadcast is the characteristic function of an independent set.
This implies that such a broadcast always exists for every graph, and thus that every graph admits an $\ell$-bounded independent broadcast for every $\ell\ge 1$.

Our goal in this section is to prove that almost all circulant graphs of the form
$C(n;1,a)$, $2\le a\le\frac{n}{2}$, admit a $2$-bounded optimal independent broadcast.
Considering the $\beta_b$-broadcasts used in the proofs of Theorems \ref{th:C(2a;1,a)} and~\ref{th:C(3a;1,a)}, we already have the following result.


\begin{proposition}\label{prop:2-bounded-2a-3a}
For every integer $a\ge 2$, the following holds.
\begin{enumerate}
\item $C(2a;1,a)$ admits a $2$-bounded $\beta_b$-broadcast if $a$ is odd or $a=2^p$ for some $p\ge 1$.
\item $C(3a;1,a)$ admits a $2$-bounded $\beta_b$-broadcast.
\end{enumerate}
\end{proposition}

We will now prove that every circulant graph of the form
$C(n;1,a)$, $a\ge 3$ and $n\ge 2a+2$, admits a $2$-bounded $\beta_b$-broadcast.
We first consider the case when $2a+2\le n < 3a$.

\begin{lemma}\label{lem:broadcast-at-most-2-2a+r}
If $n$, $a$ and $r$ are three integers such that $n = 2a + r$,
$3\le a\le\left\lfloor\frac{n}{2}\right\rfloor$ and $2\leq r < a$,
then $C(n;1,a)$ admits a $2$-bounded $\beta_b$-broadcast.
\end{lemma}

\begin{figure}
\begin{center}

\begin{tikzpicture}[scale=0.5]

   \POINTILLE{-1.5}{0}{-0.5}{0}
   \POINTILLE{-1.5}{-2}{-0.5}{-2}
   \POINTILLE{20.5}{0}{21.5}{0}
   \POINTILLE{20.5}{-2}{21.5}{-2}
   \LIGNE{-0.5}{0}{20.5}{0}
   \LIGNE{-0.5}{-2}{20.5}{-2}

   \foreach \k in {4} \draw[thick] (2*\k,0) to (2*\k+2*2,-2);
   \foreach \k in {4,9} \draw[thick] (2*\k,0) to (2*\k-2*3,-2);

   \foreach \k in {0,1,...,10} \SOM{2*\k}{0}{}{};
   \foreach \k in {0,1,...,10} \SOM{2*\k}{-2}{}{};

   \bSOM{2*4}{0}{}{}   \node[below] at (2*4,-0.2) {$v_i$}; \node[above] at (2*4,0.2) {$(3,4)$};

   \gSOM{2*3}{0}{1}{}
   \gSOM{2*5}{0}{1}{}
   \gSOM{2*1}{-2}{1}{}
   \gSOM{2*6}{-2}{1}{}

   \node[below] at (2*1,-2.2) {$v_{i+a}$};
   \node[below] at (2*6,-2.2) {$v_{i-a}$};

   \node[below] at (10,-4){(a) Item 1: $f(v_i)\in\{3,4\}$, $a=10$ and $r=5$ (so that $p=0$ and $n=25$)};
\end{tikzpicture}

\vskip 0.5cm
\begin{tikzpicture}[scale=0.5]

   \POINTILLE{-1.5}{0}{-0.5}{0}
   \POINTILLE{-1.5}{-2}{-0.5}{-2}
   \POINTILLE{22.5}{0}{23.5}{0}
   \POINTILLE{22.5}{-2}{23.5}{-2}
   \LIGNE{-0.5}{0}{22.5}{0}
   \LIGNE{-0.5}{-2}{22.5}{-2}

   \foreach \k in {5} \draw[thick] (2*\k,0) to (2*\k+2*3,-2);
   \foreach \k in {5} \draw[thick] (2*\k,0) to (2*\k-2*4,-2);

   \foreach \k in {0,1,...,11} \SOM{2*\k}{0}{}{};
   \foreach \k in {0,1,...,11} \SOM{2*\k}{-2}{}{};

   \bSOM{2*5}{0}{}{}   \node[below] at (2*5,-0.2) {$v_i$}; \node[above] at (2*5,0.2) {$(6)$};

   \gSOM{2*4}{0}{1}{}
   \gSOM{2*6}{0}{1}{}
   \gSOM{2*8}{0}{1}{}
   \gSOM{2*1}{-2}{1}{}
   \gSOM{2*3}{-2}{}{}  \node[above] at (2*3+0.5,-1.8) {1};
   \gSOM{2*8}{-2}{1}{}
   \gSOM{2*10}{-2}{1}{}

   \node[below] at (2*1,-2.2) {$v_{i+a}$};
   \node[below] at (2*8,-2.2) {$v_{i-a}$};

   \node[below] at (11,-4){(b) Item 1: $f(v_i)=6$, $a=8$ and $r=7$ (so that $p=2$ and $n=23$)};
\end{tikzpicture}

\vskip 0.5cm
\begin{tikzpicture}[scale=0.5]

   \POINTILLE{-1.5}{0}{-0.5}{0}
   \POINTILLE{-1.5}{-2}{-0.5}{-2}
   \POINTILLE{12.5}{0}{13.5}{0}
   \POINTILLE{12.5}{-2}{13.5}{-2}
   \LIGNE{-0.5}{0}{12.5}{0}
   \LIGNE{-0.5}{-2}{12.5}{-2}

   \foreach \k in {0,2,4} \draw[thick] (2*\k,0) to (2*\k+2*1,-2);
   \foreach \k in {2,4,6} \draw[thick] (2*\k,0) to (2*\k-2*1,-2);

   \foreach \k in {0,1,...,6} \SOM{2*\k}{0}{}{};
   \foreach \k in {0,1,...,6} \SOM{2*\k}{-2}{}{};

   \bSOM{2*2}{0}{}{}   \node[below] at (2*2,-0.2) {$v_i$}; \node[above] at (2*2,0.2) {$(5)$};

   \gSOM{2*1}{0}{1}{}
   \gSOM{2*3}{0}{1}{}
   \gSOM{2*5}{0}{1}{}
   \gSOM{2*1}{-2}{1}{}
   \gSOM{2*3}{-2}{1}{}
   \gSOM{2*5}{-2}{1}{}

   \node[below] at (2*1,-2.2) {$v_{i+a}$};
   \node[below] at (2*3,-2.2) {$v_{i-a}$};

   \node[below] at (6,-4){(c) Item 2: $f(v_i)=5$, $a=12$ and $r=2$ (so that $p=2$ and $n=26$)};
\end{tikzpicture}

\vskip 0.5cm
\begin{tikzpicture}[scale=0.5]

   \POINTILLE{-1.5}{0}{-0.5}{0}
   \POINTILLE{-1.5}{-2}{-0.5}{-2}
   \POINTILLE{18.5}{0}{19.5}{0}
   \POINTILLE{18.5}{-2}{19.5}{-2}
   \LIGNE{-0.5}{0}{18.5}{0}
   \LIGNE{-0.5}{-2}{18.5}{-2}

   \foreach \k in {0,3,6} \draw[thick] (2*\k,0) to (2*\k+2*1,-2);
   \foreach \k in {3,6,9} \draw[thick] (2*\k,0) to (2*\k-2*2,-2);

   \foreach \k in {0,1,...,9} \SOM{2*\k}{0}{}{};
   \foreach \k in {0,1,...,9} \SOM{2*\k}{-2}{}{};

   \bSOM{2*3}{0}{}{}   \node[below] at (2*3,-0.2) {$v_i$}; \node[above] at (2*3,0.2) {$(6)$};

   \gSOM{2*2}{0}{1}{}
   \gSOM{2*4}{0}{1}{}
   \gSOM{2*7}{0}{1}{}
   \gSOM{2*9}{0}{1}{}
   \gSOM{2*1}{-2}{1}{}
   \gSOM{2*4}{-2}{1}{}
   \gSOM{2*6}{-2}{1}{}

   \node[below] at (2*1,-2.2) {$v_{i+a}$};
   \node[below] at (2*4,-2.2) {$v_{i-a}$};

   \node[below] at (9,-4){(d) Item 3: $f(v_i)=6$, $a=12$ and $r=3$ (so that $d=2$ and $n=27$)};
\end{tikzpicture}

\vskip 0.5cm
\begin{tikzpicture}[scale=0.5]

   \POINTILLE{-1.5}{0}{-0.5}{0}
   \POINTILLE{-1.5}{-2}{-0.5}{-2}
   \POINTILLE{20.5}{0}{21.5}{0}
   \POINTILLE{20.5}{-2}{21.5}{-2}
   \LIGNE{-0.5}{0}{20.5}{0}
   \LIGNE{-0.5}{-2}{20.5}{-2}

   \foreach \k in {3,6,9} \draw[thick] (2*\k,0) to (2*\k-2*2,-2);
   \foreach \k in {3,6,9} \draw[thick] (2*\k,0) to (2*\k+2*1,-2);

   \foreach \k in {0,1,...,10} \SOM{2*\k}{0}{}{};
   \foreach \k in {0,1,...,10} \SOM{2*\k}{-2}{}{};

   \bSOM{2*3}{0}{}{}   \node[below] at (2*3,-0.2) {$v_i$}; \node[above] at (2*3,0.2) {$(5)$};

   \gSOM{2*2}{0}{1}{}
   \gSOM{2*4}{0}{1}{}
   \gSOM{2*7}{0}{1}{}
   \gSOM{2*1}{-2}{1}{}
   \gSOM{2*4}{-2}{1}{}
   \gSOM{2*6}{-2}{1}{}

   \node[below] at (2*1,-2.2) {$v_{i+a}$};
   \node[below] at (2*4,-2.2) {$v_{i-a}$};

   \node[below] at (10,-4){(e) Item 3: $f(v_i)=5$, $a=12$ and $r=3$ (so that $d=1$ and $n=27$)};
\end{tikzpicture}

\caption{\label{fig:mapping-g-2a+r}Construction of the mapping $g$ in the proof of Lemma~\ref{lem:broadcast-at-most-2-2a+r}.}
\end{center}
\end{figure}

\begin{proof}
Note that it is enough to prove that for every independent broadcast $f$ on $C(n;1,a)$,
there exists an independent broadcast $g$ on $C(n;1,a)$ such that $\sigma(g)\ge\sigma(f)$
and $g(v)\leq 2$ for every vertex $v \in V_{g}^+$.

Let $f$ be any independent broadcast on $C(n;1,a)$, and $g$ be the mapping
from $V(C(n;1,a))$ to $\{0,1,2\}$ defined as follows
(the construction of the mapping $g$ is illustrated in Figure~\ref{fig:mapping-g-2a+r},
where the value of $f(v_i)$ is indicated in brackets;
not all $a$-edges are drawn, but the missing $a$-edges are parallel to the
drawn ones;  note also that $v_{i-a}=v_{i+a+r}$, and thus $v_{i+a}$ and $v_{i-a}$
are separated by $r-1$ vertices).

\begin{enumerate}
\item If $v_i$ is an $f$-broadcast vertex such that $2 < f(v_{i})\leq r$, then we let
$$g(v_j) =
\left\{
   \begin{array}{ll}
   0 & \text{if } j=i, \\ [1ex]
   1 & \text{if $j=i-a$ }, \\ [1ex]

     & \text{or $i-1 \le j\le i+p+1$ and $j-i+1$ is even}, \\ [1ex]

    & \text{or $i+a \le j\le i+a+p$ and $j-i-a$ is even},
   \end{array}
\right.
$$
where $p=f(v_i)-3$ if $f(v_i)$ is odd, and $p=f(v_i)-4$ if $f(v_i)$ is even
(see Figure~\ref{fig:mapping-g-2a+r}(a,b)).

\item If $v_i$ is an $f$-broadcast vertex such that $\left \lceil \dfrac{a}{2}\right\rceil > f(v_{i}) \ge r+1$ and $r$ is even, then we let
$$g(v_j) =
\left\{
   \begin{array}{ll}
   0 & \text{if } j=i, \\ [1ex]
    1   & \text{if $i-1 \le j\le i+p+1$ and $j-i+1$ is even}, \\ [1ex]

    & \text{or $i+a \le j\le i+a+r+p$ and $j-i-a$ is even},
   \end{array}
\right.
$$
where $p=f(v_i)-3$ if $f(v_i)$ is odd, and $p=f(v_i)-4$ if $f(v_i)$ is even
(see Figure~\ref{fig:mapping-g-2a+r}(c)).

\item If $v_i$ is an $f$-broadcast vertex such that $f(v_{i}) \ge r+1$ and $r$ is odd, then we let
$$g(v_j) =
\left\{
   \begin{array}{ll}
   0 & \text{if } j=i, \\ [1ex]
   1 & \text{if } i-2\le j \le i + \left\lfloor\dfrac{d+2}{2}\right\rfloor r + (d \mod 2)\\ [1ex]
     & \hskip 0.5cm \text{and $(j-i+2)\mod (r+2)$ is odd}, \\ [1ex]
    & \text{or } i+a \le j \le i+a + \left\lceil\dfrac{d+2}{2}\right\rceil r + 1-(d \mod 2)\\ [1ex]
     & \hskip 0.5cm \text{and $(j-i-a+3)\mod (r+2)$ is odd},
   \end{array}
\right.
$$
where $d = f(v_i) - (r+1)$
(see Figure~\ref{fig:mapping-g-2a+r}(d,e)).

\item For every other vertex $v_k$, we let $g(v_{k})=f(v_{k})$.
\end{enumerate}

Note that, in particular, $g(v_i)=f(v_i)$ for every $f$-broadcast vertex $v_i$ with $f(v_i)\le 2$.
Moreover, all vertices set to~1 in the above items are distinct from $v_i$ and at distance
not greater than $f(v_i)$ from $v_i$, which means that their $f$-value was~0.

\medskip

We now prove that $g$ is an independent broadcast on $C(n;1,a)$.
For that, we first prove the following claim.

\begin{claim}\label{cl:distance-2a+r}
For every vertex $v_j$ whose $g$-value is set to~1 in Item 1, 2 or~3 above, we have
$d(v_i,v_j)\le f(v_i)-2$.
\end{claim}

\begin{proof}
In Item~1, every vertex whose $g$-value is set to~1 is at distance at
most $p+1\le f(v_i)-2$ from $v_i$.
In Item~2, every vertex whose $g$-value is set to~1 is at distance at
most $\max\{p+1,r-1\}\le f(v_i)-2$ from $v_i$.

Consider now a 
vertex $v_j$ whose $g$-value is set to~1 in Item~3 and whose distance to $v_i$ is maximal
(see Figure~\ref{fig:mapping-g-2a+r}(d,e)), and
suppose first that $i-2\le j \le i + \left\lfloor\frac{d+2}{2}\right\rfloor r + (d \mod 2)$.
Since $r\ge 3$ and $r$ is odd, $v_j$ is at distance
$d(v_i,v_j) = 2\left\lfloor\frac{d+2}{2}\right\rfloor + \left\lfloor\frac{r}{2}\right\rfloor - 1$
from $v_i$ (going to $v_{i+\left\lfloor\frac{d+2}{2}\right\rfloor r}$
using $(2\left\lfloor\frac{d+2}{2}\right\rfloor)$ $a$-edges, and then
back to $v_j$ using $(\left\lfloor\frac{r}{2}\right\rfloor - 1)$ $1$-edges).
Since $d=f(v_i)-r-1$ and $r\ge 3$, this gives
$$d(v_i,v_j) \le f(v_i) - r + 1 + \frac{r-1}{2} - 1 = f(v_i) - \frac{r+1}{2} \le f(v_i) - 2.$$
Suppose finally that $i+a \le j \le i+a + \left\lceil\frac{d+2}{2}\right\rceil r + 1-(d \mod 2)$.
In that case, $v_j$ is at distance
$d(v_i,v_j) = 1 + 2(\left\lfloor\frac{d+2}{2}\right\rfloor - 1) + \left\lfloor\frac{r}{2}\right\rfloor - 1$
from $v_i$ (going to $v_{i+a + \left\lceil\frac{d+2}{2}\right\rceil r}$
using $(1 + 2(\left\lfloor\frac{d+2}{2}\right\rfloor-1))$ $a$-edges, and then
back to $v_j$ using $(\left\lfloor\frac{r}{2}\right\rfloor - 1)$ $1$-edges).
As before, since $d=f(v_i)-r-1$ and $r\ge 3$, this gives
\begin{align*}
d(v_i,v_j)
  & \le 1 + 2\left(\left\lfloor\frac{d+2}{2}\right\rfloor - 1\right) + \frac{r-1}{2} - 1
 = 2\left(\left\lfloor\frac{f(v_i) - r + 1}{2}\right\rfloor - 1\right) + \frac{r-1}{2}
  \\[2ex]
  & \le 2\left(\frac{f(v_i) - r + 1}{2} \right) + \frac{r-1}{2} - 2
  = f(v_i) - \frac{r-1}{2} - 2
  \\[2ex]
  &   \le f(v_i) - 2,
\end{align*}
which concludes the proof of the claim.
\end{proof}

Thanks to this claim, and since $f$ was an independent broadcast on $C(n;1,a)$,
no $g$-broadcast vertex $v_i$ with $g(v_i)=f(v_i)\in\{1,2\}$ $g$-dominates
a vertex whose $g$-value has been set to~1.
Therefore, in order to prove that $g$ is indeed an independent broadcast on $C(n;1,a)$,
it remains to prove that the set of vertices whose $g$-value has been set to~1 is an
independent set.
Moreover, thanks to Claim~\ref{cl:distance-2a+r}, Items 1, 2 and~3 can be considered
separately.
This is readily the case for vertices whose $g$-value has been set to~1 in Item 1 and~2.
In Item~3, thanks to the parity of their subscript, no two such vertices are linked
by a $1$-edge.
Moreover, any two such vertices cannot be linked by an $a$-edge since
$(j-i+2)-(j-i-a+3)=a-1$.

\medskip

In order to finish the proof, we only need to show that we have $\sigma(g)\ge\sigma(f)$.
Indeed, in Item~1, the number of vertices set to~1 is
$n_1 = 1 + \frac{p+4}{2} + \frac{p+2}{2} = p + 4$,
which gives $n_1 = f(v_i)$ if $f(v_i)$ is even,
or $n_1 = f(v_i) + 1 > f(v_i)$ if $f(v_i)$ is odd.

In Item~2, the number of vertices set to~1 is
$n_2 = \frac{p+4}{2} + \frac{p+2}{2} + \frac{r}{2} = \frac{2p+6+r}{2}$,
which gives $n_2 = \frac{2f(v_i)-2+r}{2}\ge f(v_i)$ (recall that $r\ge 2$)
if $f(v_i)$ is even,
or $n_2 = \frac{2f(v_i)+r}{2} > f(v_i)$ if $f(v_i)$ is odd.

Finally consider Item~3.
Observe that, since $r+2$ is odd, in every sequence of $r+2$ consecutive
vertices lying between $v_{i-2}$ and $v_{i + \left\lfloor\frac{d+2}{2}\right\rfloor r + (d \mod 2)}$,
or between $v_{i+a}$ and $v_{i+a + \left\lceil\frac{d+2}{2}\right\rceil r + 1-(d \mod 2)}$,
exactly $\frac{r+1}{2}$ vertices are set to~$1$.
Therefore, the number of vertices set to~1 in Item~3 is either
$$n_3 = \left\lceil \frac{r+1}{2(r+2)} \left(\frac{d+2}{2}r+3 \right) \right\rceil
+  \left\lceil \frac{r+1}{2(r+2)} \left(\frac{d+2}{2}r+2 \right) \right\rceil,$$
if $d$ is even, or
$$n_3 = \left\lceil \frac{r+1}{2(r+2)} \left(\frac{d+1}{2}r+4 \right) \right\rceil
+ \left\lceil  \frac{r+1}{2(r+2)} \left(\frac{d+3}{2}r+1 \right) \right\rceil,$$
if $d$ is odd.
In both cases, we get
$$n_3 \ge \frac{r+1}{2(r+2)} \big((d+2)r+5 \big)
= r+1 + \dfrac{(dr + 1)(r+1)}{2(r+2)} \geq r+1 + \dfrac{dr(r+1)}{2(r+2)}.$$
Since $r\ge 3$, we have $r(r+1) \ge 2(r+2)$, and thus
$$n_3 \ge r + 1 + d = f(v_i).$$

We thus have $\sigma(g)\ge\sigma(f)$, as required.
This completes the proof.
\end{proof}


We now consider the case $n > 3a$.

\begin{lemma}\label{lem:broadcast-at-most-2-general}
If $n$ and $a$ are two integers such that $3\le a\le\left\lfloor\frac{n}{2}\right\rfloor$ and $n>3a$,
then $C(n;1,a)$ admits a $2$-bounded $\beta_b$-broadcast.
\end{lemma}

\begin{figure}
\begin{center}
\begin{tikzpicture}[scale=0.5]
   \node[above] at (0,0.2){$f$};
   \node[below] at (0,-0.2){$g$};

   \POINTILLE{0}{0}{1}{0}   \LIGNE{1}{0}{24}{0}   \POINTILLE{24}{0}{25}{0}

   \foreach \k in {2,9,16} \draw[thick] (\k,0) .. controls (\k+1,2) and (\k+6,2) .. (\k+7,0);
   \foreach \k in {1,2,...,24} \SOM{\k}{0}{}{};

   \bSOM{9}{0}{}{}   \node[above] at (9,0.5) {5,6};

   \gSOM{2}{0}{}{1}
   \gSOM{4}{0}{}{1}
   \gSOM{8}{0}{}{1}
   \gSOM{10}{0}{}{1}
   \gSOM{12}{0}{}{1}
   \gSOM{16}{0}{}{1}
   \gSOM{18}{0}{}{1}

   \node[below] at (12.5,-2){(a) Item 1: $f(v_i)\in\{5,6\}$ and $a=7$ (so that $p=2$)};
\end{tikzpicture}

\vskip 0.5cm
\begin{tikzpicture}[scale=0.5]
   \node[above] at (0,0.2){$f$};
   \node[below] at (0,-0.2){$g$};

   \POINTILLE{0}{0}{1}{0}   \LIGNE{1}{0}{18}{0}   \POINTILLE{18}{0}{19}{0}

   \foreach \k in {2,7,12} \draw[thick] (\k,0) .. controls (\k+1,1.8) and (\k+4,1.8) .. (\k+5,0);
   \foreach \k in {1,2,...,18} \SOM{\k}{0}{}{};

   \gSOM{2}{0}{}{1}
   \gSOM{4}{0}{}{1}
   \gSOM{6}{0}{}{1}   \bSOM{7}{0}{7}{}
   \gSOM{8}{0}{}{1}
   \gSOM{10}{0}{}{1}
   \gSOM{12}{0}{}{1}
   \gSOM{14}{0}{}{1}
   \gSOM{16}{0}{}{1}

   \node[below] at (9,-2){(b) Item 2: $f(v_i)=7$ and $a=5$ (so that $d=1$)};
\end{tikzpicture}

\vskip 0.5cm
\begin{tikzpicture}[scale=0.5]
   \node[above] at (0,0.2){$f$};
   \node[below] at (0,-0.2){$g$};

   \POINTILLE{0}{0}{1}{0}   \LIGNE{1}{0}{23}{0}   \POINTILLE{23}{0}{24}{0}

   \foreach \k in {2,6,10,14,18} \draw[thick] (\k,0) .. controls (\k+1,1.5) and (\k+3,1.5) .. (\k+4,0);
   \foreach \k in {1,2,...,23} \SOM{\k}{0}{}{};

   \gSOM{2}{0}{}{1}
   \gSOM{4}{0}{}{1}       \bSOM{6}{0}{7}{}
   \gSOM{7}{0}{}{1}
   \gSOM{9}{0}{}{1}
   \gSOM{12}{0}{}{1}
   \gSOM{14}{0}{}{1}
   \gSOM{17}{0}{}{1}
   \gSOM{19}{0}{}{1}
   \gSOM{22}{0}{}{1}

   \node[below] at (12,-2){(c) Item 3: $f(v_i)=7$ and $a=4$ (so that $d=2$)};
\end{tikzpicture}

\caption{\label{fig:mapping-g}Construction of the mapping $g$ in the proof of Lemma~\ref{lem:broadcast-at-most-2-general}.}
\end{center}
\end{figure}

\begin{proof}
Again, it is enough to prove that for every independent broadcast $f$ on $C(n;1,a)$,
there exists an independent broadcast $g$ on $C(n;1,a)$ such that $\sigma(g)\ge\sigma(f)$
and $g(v)\leq 2$ for every vertex $v \in V_{g}^+$.

Let $f$ be any independent broadcast on $C(n;1,a)$, and $g$ be the mapping
from $V(C(n;1,a))$ to $\{0,1,2\}$ defined as follows
(the construction of the mapping $g$ is illustrated in Figure~\ref{fig:mapping-g},
not all $a$-edges being drawn).

\begin{enumerate}

\item If $v_i$ is an $f$-broadcast vertex such that $2 < f(v_{i})\leq a$, then we let
$$g(v_j) =
\left\{
   \begin{array}{ll}
   0 & \text{if } j=i, \\ [1ex]
   1 & \text{if $i-1 \le j\le i+p+1$ and $j-i+1$ is even,} \\ [1ex]
     & \text{or $i-a\le j\le i-a+p$ and $j-i+a$ is even,} \\ [1ex]
     &\text{or $i+a\le j\le i+a+p$ and $j-i-a$ is even},
   \end{array}
\right.
$$
where $p=f(v_i)-3$ if $f(v_i)$ is odd, and $p=f(v_i)-4$ if $f(v_i)$ is even
(see Figure~\ref{fig:mapping-g}(a)).

\item If $v_i$ is an $f$-broadcast vertex such that $f(v_{i}) \ge a+1$ and $a$ is odd, then we let
$$g(v_j) =
\left\{
   \begin{array}{ll}
   0 & \text{if } j=i, \\ [1ex]
   1 & \text{if } i-a \le j \le i+(1+d)a\ \text{and $j-i+a$ is even},
   \end{array}
\right.
$$
where $d = f(v_i) - (a+1)$
(see Figure~\ref{fig:mapping-g}(b)).

\item If $v_i$ is an $f$-broadcast vertex such that $f(v_{i}) \ge a+1$ and $a$ is even, then we let
$$g(v_j) =
\left\{
   \begin{array}{ll}
   0 & \text{if } j=i \\ [1ex]
   1 & \text{if } i-a-1 \le j \le i+(2+d)a\ \text{and $(j-i+a+1) \mod (a+1)$ is odd},
   \end{array}
\right.
$$
where $d = f(v_i) - (a+1)$
(see Figure~\ref{fig:mapping-g}(c)).

\item For every other vertex $v_k$, we let $g(v_{k})=f(v_{k})$.
\end{enumerate}

Note that, as in the proof of the previous lemma,
$g(v_i)=f(v_i)$ for every $f$-broadcast vertex $v_i$ with $f(v_i)\le 2$.
Moreover, all vertices set to~1 in the above items are also
distinct from $v_i$ and at distance
not greater than $f(v_i)$ from $v_i$, which means that their $f$-value was~0.

\medskip

We now prove that $g$ is an independent broadcast on $C(n;1,a)$.
For that, we first prove the following claim.

\begin{claim}\label{cl:distance-general}
For every vertex $v_j$ whose $g$-value is set to~1 in Item 1, 2 or~3 above, we have
$d(v_i,v_j)\le f(v_i)-2$.
\end{claim}

\begin{proof}
In Item~1, every vertex whose $g$-value is set to~1 is at distance at most $p+1\le f(v_i)-2$ from $v_i$.

Among the vertices whose $g$-value might be set to~1 in Item~2, the
vertex whose distance to $v_i$ is maximal is, since $a$ is odd,
the vertex $v_j$ with $j=da+\frac{a+1}{2}$, which gives
$$d(v_i,v_j)=d+\frac{a+1}{2}=  f(v_i)-\frac{a+1}{2} \leq f(v_i)-2$$
(recall that, in that case, we have  $a\geq 3$).

Similarly, among the vertices whose $g$-value might be set to~1 in Item~3,
the vertex whose distance to $v_i$ is maximal is, since $a$ is even,
the vertex $v_j$ with $j=(d+1)a+\frac{a}{2}$, which gives
$$d(v_i,v_j)=d+1+\frac{a}{2}=f(v_i)-\frac{a}{2}  \leq f(v_i)-2$$
(recall that, in that case, we have $a\ge 4$).

This concludes the proof of the claim.
\end{proof}

Thanks to this claim, and since $f$ was an independent broadcast on $C(n;1,a)$,
no $g$-broadcast vertex $v_i$ with $g(v_i)=f(v_i)\in\{1,2\}$ $g$-dominates
a vertex whose $g$-value has been set to~1.
Therefore, in order to prove that $g$ is indeed an independent broadcast on $C(n;1,a)$,
it remains to prove that the set of vertices whose $g$-value has been set to~1 is an
independent set.
Moreover, thanks to Claim~\ref{cl:distance-general}, Items 1, 2 and~3 can be considered
separately.
This is readily the case for vertices whose $g$-value has been set to~1 in Item~1.
It follows from the parity of their subscript in Item~2 (neither a $1$-edge nor
an $a$-edge, since $a$ is odd, can link any two such vertices),
and from the value modulo $(a+1)$
of their subscript in Items~3 (which, again, implies that neither a $1$-edge nor
an $a$-edge, since $a$ is even, can link any two such vertices).

\medskip

In order to finish the proof, we only need to show that we have $\sigma(g)\ge\sigma(f)$.
Indeed, in Item~1, the number of vertices set to~1 is
$n_1 = \frac{p+4}{2} + \frac{p+2}{2} + \frac{p+2}{2} = \frac{3p+8}{2}$,
which gives $n_1 = \frac{3f(v_i)-4}{2}\ge f(v_i)$ if $f(v_i)$ is even (in that case, $f(v_i)\ge 4$, and
equality holds only when $f(v_i)=4$),
or $n_1 = \frac{3f(v_i)-1}{2} > f(v_i)$ if $f(v_i)$ is odd (in that case, $f(v_i)\ge 3$).

In Item~2, the number of vertices set to~1 is
$n_2 = \left\lceil\frac{(2+d)a+1}{2}\right\rceil = \left\lceil\frac{2a+1+da}{2}\right\rceil$.
If $d=0$, since $a$ is odd, we get $n_2 = \left\lceil\frac{2a+1}{2}\right\rceil = a+1=f(v_i)$.
Otherwise, that is, if $d\ge 1$, since $a\ge 3$,
we get $n_2 = \left\lceil\frac{2a+1+da}{2}\right\rceil
\ge a+\frac{1}{2}+\frac{ad}{2}\ge a+d+1 = f(v_i)$.


Finally, in Item~3, note that, for every sequence of $a+1$ consecutive vertices, $\frac{a}{2}$
of them are set to~1.
Therefore, the total number of vertices set to~1 is
$$n_3 = \left\lceil \frac{a}{2(a+1)}\left((d+3)a+2\right)\right\rceil
\ge \frac{3a^2+da^2+2a}{2(a+1)}
= a + \frac{da^2}{2(a+1)} + \frac{a^2}{2(a+1)} .$$
Since $a\ge 4$, we have $\frac{a^2}{2(a+1)}>1$, which gives $n_3>a+d+1=f(v_i)$.

We thus have $\sigma(g)\ge\sigma(f)$, as required.
This completes the proof.
\end{proof}


From Proposition~\ref{prop:2-bounded-2a-3a} and Lemmas \ref{lem:broadcast-at-most-2-2a+r}
and~\ref{lem:broadcast-at-most-2-general}, we directly get the following theorem.

\begin{theorem}\label{th:2-bounded}
Every circulant graph of the form $C(n;1,a)$, $3\le a\le\left\lfloor\frac{n}{2}\right\rfloor$,
admits a $2$-bounded $\beta_b$-broadcast if none of the following conditions is satisfied: (i) $n=2a$ and $a$ is even, or (ii) $n=2a+1$.
\end{theorem}

The following example will show that when $a=2$ or $n=2a+1$, not all circulant graphs of the form $C(n;1,a)$ admit a $2$-bounded $\beta_b$-broadcast.
Consider  the circulant graph  $C(21;1,2)$, 
and let $f$ be the mapping from $V(C(21;1,2))$ to $\{0,3\}$ defined by $f(v_0)=f(v_7)=f(v_{14})=3$ and $f(v_i)=0$ otherwise.
Since $2$ is even,  $f$ is clearly an independent broadcast on $C(21;1,2)$, with cost $\sigma(f)=9$, and thus $9 \leq \beta_b(C(21;1,2))$.
Now, suppose that there exists a $2$-bounded $\beta_b$-broadcast $g$ on $C(21;1,2)$.
If $|V_g^+| \le 4$, we  immediately get $\sigma(g) \leq 8$, since $g$ is $2$-bounded.
Suppose now $|V_g^+| > 4$.
Each vertex $v_i\in V_g^+$ dominates at most three vertices among $\{v_i,v_{i+1},v_{i+2}\}$ (subscripts are taken modulo~$21$), and none of these vertices is dominated more than once.
Therefore, since $g$ is an independent broadcast, we get
$$\sum_{v_i\in V_g^+}(1+2g(v_i)) = |V_g^+| + 2g(V_g^+) \le 21,$$
which gives (recall that $|V_g^+| > 4$)
$$\sigma(g) =  g(V_g^+) \leq  \frac{21 - |V_g^+|}{2} \le 8.$$
In both cases, we get a contradiction to the optimality of  $g$.
Finally, since $C(21;1,10)$ is isomorphic to  $C(21;1,2)$,  we get that
there also exist circulant graphs of the form $C(2a+1;1,a)$ that do not admit any $2$-bounded $\beta_b$- broadcast.

\section{General bounds on the independence broadcast number of \texorpdfstring{$C(n;1,a)$}{C(n;1,a)}}\label{sec:bounds}

In this section, we will provide some general upper and lower bounds on the cost of independent broadcasts on
circulant graphs of the form $C(n;1,a)$, $2\le a\le\left\lfloor\frac{n}{2}\right\rfloor$, that will be useful in the next section.

We first introduce some notation and a useful lemma.
%
Let $f$ be an independent broadcast on $C(n;1,a)$.
We then let
$$V_f^1=\{v_i\in V_f^+\ |\ f(v_i) = 1\},\  V_f^{2}=\{v_j\in V_f^+\ |\ f(v_j) = 2\} \mbox{ and } V_f^{\ge 2}=\{v_j\in V_f^+\ |\ f(v_j) \ge 2\}.$$
In particular, if $f$ is $2$-bounded, we then have $V_f^+ = V_f^1 \cup V_f^{2}$.

Consider now a $2$-bounded independent broadcast $f$ and  any vertex $v_i\in V_f^1$ such that $f(v_{i-1}) = f(v_{i-2}) = 0$.
Since $f$ is an independent broadcast, we necessarily have $f(v_{i+1}) = 0$.
Moreover, we then have either $f(v_{i+2}) = 0$ or $f(v_{i+2}) = 1$.
Therefore, the broadcast values of the sequence of vertices $v_iv_{i+1}v_{i+2}\dots$
is of the form either $100$, $10100$ or $1010\dots100$.

For each vertex $v_i\in V_f^1$ such that $f(v_{i-1}) = f(v_{i-2}) = 0$,
we then let
$$A_f^i = \{v_{i+\ell},\ 0\le \ell\le 2p+2\}$$
be the set of vertices satisfying
(i) $f(v_{i+2k})=1$ and $f(v_{i+2k+1})=0$ for every $k$, $0\le k\le p$,
and (ii) $f(v_{i+2p+2})=0$.

Now, for each vertex $v_j\in V_f^{2}$, we let
$$B_f^j =  \{v_{j-a+1}\}
  \cup  \{v_j,v_{j+1},v_{j+2}\}
  \cup \{v_{j+a+1}\}.$$

\begin{figure}
\begin{center}
 \begin{tikzpicture}[scale=0.5]

 \POINTILLE{0}{0}{1}{0}   \LIGNE{1}{0}{18}{0}   \POINTILLE{18}{0}{19}{0}

   \foreach \k in {3,10} \draw[thick] (\k,0) .. controls (\k+1,2) and (\k+6,2) .. (\k+7,0);
   \foreach \k in {1,2,...,18} \SOM{\k}{0}{}{};

   \bSOM{9}{0}{}{1}

   \gSOM{1}{0}{}{0}
   \gSOM{2}{0}{}{0}
   \bSOM{3}{0}{$v_i$}{1}
   \gSOM{4}{0}{}{0}
   \bSOM{5}{0}{}{1}
   \gSOM{6}{0}{}{0}
   \bSOM{7}{0}{}{1}
   \gSOM{8}{0}{}{0}
   \gSOM{10}{0}{}{0}
   \gSOM{11}{0}{}{0}

     \LIGNE{3}{-1.2}{3}{-1.5}
   \LIGNE{3}{-1.5}{11}{-1.5}
   \LIGNE{11}{-1.5}{11}{-1.2}
 \node[below] at (7,-1.6){ $A_f^i$  } ;
   \node[below] at (10,-2.8){The vertices of the set $A_f^i$  , $a=7$ } ;
\end{tikzpicture}

\begin{tikzpicture}[scale=0.5]
   \POINTILLE{0}{0}{1}{0}   \LIGNE{1}{0}{24}{0}   \POINTILLE{24}{0}{25}{0}

   \foreach \k in {2,9,16} \draw[thick] (\k,0) .. controls (\k+1,2) and (\k+6,2) .. (\k+7,0);
   \foreach \k in {1,2,...,24} \SOM{\k}{0}{}{};

   \bSOM{9}{0}{$v_j$}{2}

   \gSOM{1}{0}{}{}
   \gSOM{2}{0}{}{}
   \bSOM{3}{0}{}{}
   \gSOM{7}{0}{}{}
   \gSOM{8}{0}{}{}
   \bSOM{10}{0}{}{}
   \bSOM{11}{0}{}{}
   \gSOM{15}{0}{}{}
   \gSOM{16}{0}{}{}
   \bSOM{17}{0}{}{}
   \gSOM{23}{0}{}{}

   \node[below] at (12.5,-2){The vertices of the set $B_f^j$ (the black vertices) , $a=7$ } ;
\end{tikzpicture}
\caption{\label{fig:AiBj}The sets $A_f^i$ and $B_f^j$.}
\end{center}
\end{figure}

The definition of these two sets is illustrated in Figure~\ref{fig:AiBj}.
These sets have the following properties.


\begin{lemma}\label{lem:AiBj}
For every $2$-bounded independent broadcast $f$ on $C(n;1,a)$, $2\le a\le \left\lfloor\frac{n}{2}\right\rfloor$, the following holds.
\begin{enumerate}
\item For every vertex $v_i\in V_f^1$, $|A_f^i| = 2f(A_f^i)+1$.
\item For every vertex $v_j\in V_f^{2}$, $|B_f^j| = 5$.
\item $\sum_{v_i\in V_f^1}|A_f^i|  +  \sum_{v_j\in V_f^{2}}|B_f^j| \le n$.
\end{enumerate}
\end{lemma}

\begin{proof}
The first two items directly follow from the definition of the sets $A_f^i$ and $B_f^j$.
It also follows from the definition that $A_f^i \cap A_f^{i'} = \emptyset$ for every
two distinct vertices $v_i$ and $v_{i'}$ in $V_f^1$.
Similarly, we necessarily have  $B_f^j \cap B_f^{j'} = \emptyset$
for every two distinct vertices $v_j$ and $v_{j'}$ in $V_f^2$,
since otherwise we would have $d(v_{j},v_{j'})\leq \max \{f(v_{j}),f(v_{j'})\} = 2$,
contradicting the fact that $f$ is an independent broadcast.
The same argument gives $A_f^i \cap B_f^{j} = \emptyset$
for every two vertices $v_i\in V_f^1$ and $v_j\in V_f^{2}$.
All together, these three properties imply that Item~3 also holds.
\end{proof}

The next result provides a general upper bound on the broadcast
independence number of circulant graphs of the form $C(n;1,a)$,
with $3\le a\le\left\lfloor\frac{n}{2}\right\rfloor$ and $3a\leq n$.

\begin{proposition}\label{prop:bound1}
If $n$ and $a$ are two integers such that
$3\le a\le\left\lfloor\frac{n}{2}\right\rfloor$ and $3a\leq n$,
then, for every $2$-bounded independent broadcast $f$ on $C(n;1,a)$, we have
$$\sigma(f)\leq \left\lfloor\frac{n-\left|V_f^{2}\right|}{2}\right\rfloor.$$
\end{proposition}

\begin{proof}
Let $f$ be any $2$-bounded independent broadcast on $C(n;1,a)$ (it follows from Theorem~\ref{th:2-bounded} that such broadcasts exist).
From Lemma~\ref{lem:AiBj}, we get
$$\sum_{v_i\in V_f^1}|A_f^i|  +  \sum_{v_j\in V_f^{ 2}}|B_f^j|
   =  2f(V_f^1) + \left|V_f^1\right| + 3f(V_f^{ 2}) - \left|V_f^{ 2}\right| \le n,$$
which gives
$$ 2f(V_f^+) = 2f(V_f^1) + 2f(V_f^{ 2}) \le n - \left|V_f^1\right| + \left|V_f^{2}\right| - f(V_f^{ 2})
\le  n + \left|V_f^{ 2}\right| - f(V_f^{ 2}).$$
Now, since $f(v_j)= 2$ for every $v_j\in V_f^{ 2}$,
we have $f(V_f^{ 2}) = 2\left|V_f^{ 2}\right|$, and thus
$$\sigma(f) = f(V_f^+) \leq \left\lfloor\frac{n-\left|V_f^{2}\right|}{2}\right\rfloor.$$
This completes the proof.
\end{proof}

When $a$ is even, the upper bound given in Proposition~\ref{prop:bound1} can be
improved as follows.

\begin{proposition}\label{prop:bound2}
If $n$ and $a$ are two integers such that
$2\le a\le\left\lfloor\frac{n}{2}\right\rfloor$ and $a$ is even, then,
for every $2$-bounded independent broadcast $f$ on $C(n;1,a)$, we have
$$\sigma(f)\leq \left\lfloor\frac{a}{2(a+1)}\left(n-\frac{a-4}{a}\left|V_f^{2}\right|\right) \right\rfloor.$$
\end{proposition}

\begin{proof}
Let $f$ be any $2$-bounded independent broadcast on $C(n;1,a)$.
Observe first that we necessarily have $|A_f^i|\le a+1$ for every vertex $v_i\in V_f^1$,
since otherwise this would give $f(v_i)=f(v_{i+a})=1$, contradicting the fact
that $f$ is an independent broadcast.
This implies $f(A_f^i)\le \frac{a}{2}$.
Using item~1 of Lemma~\ref{lem:AiBj}, we then get
$$\frac{|A_f^i|}{f(A_f^i)} = \frac{2f(A_f^i)+1}{f(A_f^i)} = 2 + \frac{1}{f(A_f^i)}
\ge 2 + \frac{2}{a} = \frac{2(a+1)}{a},$$
and thus
$$|A_f^i| \ge \frac{2(a+1)}{a}f(A_f^i).$$

From Lemma~\ref{lem:AiBj}, we then get
$$n \ge \sum_{v_i\in V_f^1}|A_f^i|  +  \sum_{v_j\in V_f^{ 2}}|B_f^j|
   \ge  \frac{2(a+1)}{a}f(V_f^1) + 3f(V_f^{ 2}) - \left|V_f^{2}\right|,$$
which gives
$$n \ge  \frac{2(a+1)}{a}f(V)  + \frac{a-2}{a}f(V_f^{ 2}) - \left|V_f^{ 2}\right|.$$

Finally, since $f(V_f^{ 2}) = 2\left|V_f^{ 2}\right|$, we get
$$\frac{2(a+1)}{a}f(V) \le n + \left|V_f^{ 2}\right| - \frac{2(a-2)}{a}\left|V_f^{ 2}\right|
= n - \frac{a-4}{a}\left|V_f^{ 2}\right|,$$
and thus
$$\sigma(f) = f(V) \le \left\lfloor\frac{a}{2(a+1)}\left(n-\frac{a-4}{a}\left|V_f^{2}\right|\right) \right\rfloor.$$
This completes the proof.
\end{proof}

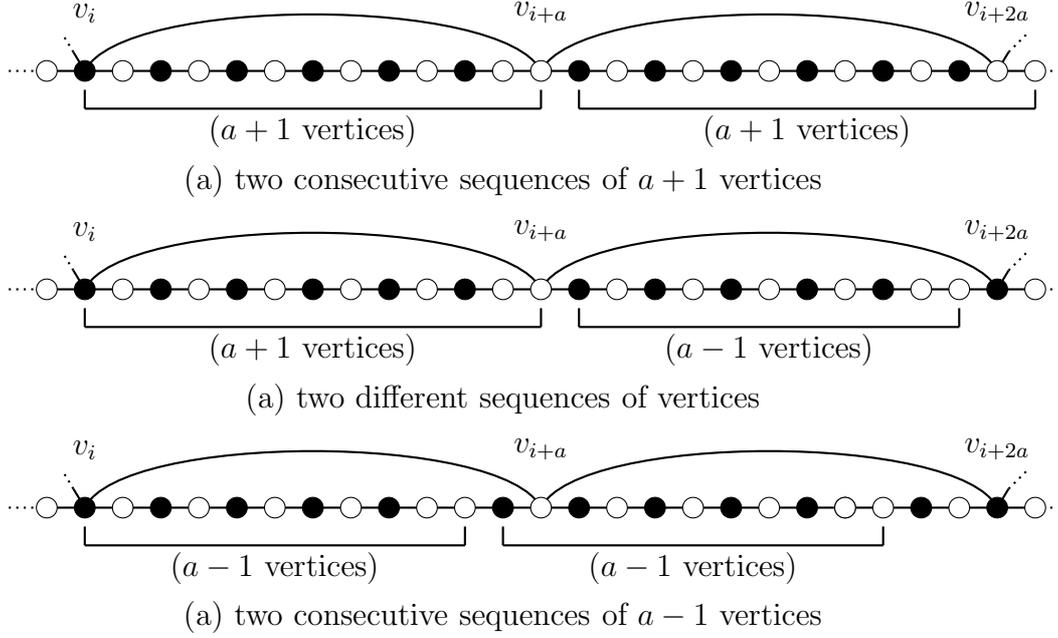
\begin{figure}
\begin{center}
\begin{tikzpicture}[scale=0.5]
   \POINTILLE{-1}{0}{0}{0}
   \LIGNE{0}{0}{26}{0}
    \POINTILLE{26}{0}{26.5}{0}

   \LIGNE{1}{0}{0.7}{0.5}
   \POINTILLE{0.7}{0.5}{0.4}{1}
   \LIGNE{25}{0}{25.3}{0.5}
   \POINTILLE{25.3}{0.5}{25.8}{1}

   \LIGNE{1}{-0.5}{1}{-1}
   \LIGNE{1}{-1}{13}{-1}
   \LIGNE{13}{-1}{13}{-0.5}

   \LIGNE{14}{-0.5}{14}{-1}
   \LIGNE{14}{-1}{26}{-1}
   \LIGNE{26}{-1}{26}{-0.5}

   \foreach \k in {1,13}
      \draw[thick] (\k,0) .. controls (\k+1,2) and (\k+11,2) .. (\k+12,0);

   \SOM{0}{0}{}{}
   \bSOM{1}{0}{}{}
   \SOM{2}{0}{}{}
   \bSOM{3}{0}{}{}
   \SOM{4}{0}{}{}
   \bSOM{5}{0}{}{}
   \SOM{6}{0}{}{}
   \bSOM{7}{0}{}{}
   \SOM{8}{0}{}{}
   \bSOM{9}{0}{}{}
   \SOM{10}{0}{}{}
   \bSOM{11}{0}{}{}
   \SOM{12}{0}{}{}
   \SOM{13}{0}{}{}

   \bSOM{14}{0}{}{}
   \SOM{15}{0}{}{}
   \bSOM{16}{0}{}{}
   \SOM{17}{0}{}{}
   \bSOM{18}{0}{}{}
   \SOM{19}{0}{}{}
   \bSOM{20}{0}{}{}
   \SOM{21}{0}{}{}
   \bSOM{22}{0}{}{}
   \SOM{23}{0}{}{}
   \bSOM{24}{0}{}{}
   \SOM{25}{0}{}{}
   \SOM{26}{0}{}{}

\tobedone

   \node[above] at (1,1) {$v_{i}$};
   \node[above] at (13,1) {$v_{i+a}$};
   \node[above] at (25,1) {$v_{i+2a}$};

   \node[below] at (7,-0.9) {($a+1$ vertices)};
   \node[below] at (20,-0.9) {($a+1$ vertices)};

  \node[below] at (12,-2.2){(a) two consecutive sequences of $a+1$ vertices };

\end{tikzpicture}

  \begin{tikzpicture}[scale=0.5]
   \POINTILLE{-1}{0}{0}{0}
   \LIGNE{0}{0}{26}{0}
    \POINTILLE{26}{0}{26.5}{0}

   \LIGNE{1}{0}{0.7}{0.5}
   \POINTILLE{0.7}{0.5}{0.4}{1}
   \LIGNE{25}{0}{25.3}{0.5}
   \POINTILLE{25.3}{0.5}{25.8}{1}

   \LIGNE{1}{-0.5}{1}{-1}
   \LIGNE{1}{-1}{13}{-1}
   \LIGNE{13}{-1}{13}{-0.5}

   \LIGNE{14}{-0.5}{14}{-1}
   \LIGNE{14}{-1}{24}{-1}
   \LIGNE{24}{-1}{24}{-0.5}

   \foreach \k in {1,13}
      \draw[thick] (\k,0) .. controls (\k+1,2) and (\k+11,2) .. (\k+12,0);

   \SOM{0}{0}{}{}
   \bSOM{1}{0}{}{}
   \SOM{2}{0}{}{}
   \bSOM{3}{0}{}{}
   \SOM{4}{0}{}{}
   \bSOM{5}{0}{}{}
   \SOM{6}{0}{}{}
   \bSOM{7}{0}{}{}
   \SOM{8}{0}{}{}
   \bSOM{9}{0}{}{}
   \SOM{10}{0}{}{}
   \bSOM{11}{0}{}{}
   \SOM{12}{0}{}{}
   \SOM{13}{0}{}{}

   \bSOM{14}{0}{}{}
   \SOM{15}{0}{}{}
   \bSOM{16}{0}{}{}
   \SOM{17}{0}{}{}
   \bSOM{18}{0}{}{}
    \SOM{19}{0}{}{}
   \bSOM{20}{0}{}{}
   \SOM{21}{0}{}{}
   \bSOM{22}{0}{}{}
   \SOM{23}{0}{}{}
   \SOM{24}{0}{}{}
   \bSOM{25}{0}{}{}
   \SOM{26}{0}{}{}

\tobedone

  \node[above] at (1,1) {$v_{i}$};
   \node[above] at (13,1) {$v_{i+a}$};
   \node[above] at (25,1) {$v_{i+2a}$};

   \node[below] at (7,-0.9) {($a+1$ vertices)};
   \node[below] at (19,-0.9) {($a-1$ vertices)};

\node[below] at (12,-2.2){(a) two  different sequences of vertices };

\end{tikzpicture}

 \begin{tikzpicture}[scale=0.5]
    \POINTILLE{-1}{0}{0}{0}
   \LIGNE{0}{0}{26}{0}
    \POINTILLE{26}{0}{26.5}{0}

   \LIGNE{1}{0}{0.7}{0.5}
   \POINTILLE{0.7}{0.5}{0.4}{1}
   \LIGNE{25}{0}{25.3}{0.5}
   \POINTILLE{25.3}{0.5}{25.8}{1}

   \LIGNE{1}{-0.5}{1}{-1}
   \LIGNE{1}{-1}{11}{-1}
   \LIGNE{11}{-1}{11}{-0.5}

   \LIGNE{12}{-0.5}{12}{-1}
   \LIGNE{12}{-1}{22}{-1}
   \LIGNE{22}{-1}{22}{-0.5}

   \foreach \k in {1,13}
      \draw[thick] (\k,0) .. controls (\k+1,2) and (\k+11,2) .. (\k+12,0);

   \SOM{0}{0}{}{}
   \bSOM{1}{0}{}{}
   \SOM{2}{0}{}{}
   \bSOM{3}{0}{}{}
   \SOM{4}{0}{}{}
   \bSOM{5}{0}{}{}
   \SOM{6}{0}{}{}
   \bSOM{7}{0}{}{}
   \SOM{8}{0}{}{}
   \bSOM{9}{0}{}{}
   \SOM{10}{0}{}{}
   \SOM{11}{0}{}{}
   \bSOM{12}{0}{}{}
   \SOM{13}{0}{}{}

   \bSOM{14}{0}{}{}
   \SOM{15}{0}{}{}
   \bSOM{16}{0}{}{}
   \SOM{17}{0}{}{}
   \bSOM{18}{0}{}{}
      \SOM{19}{0}{}{}
   \bSOM{20}{0}{}{}
   \SOM{21}{0}{}{}
   \SOM{21}{0}{}{}
   \SOM{22}{0}{}{}
   \bSOM{23}{0}{}{}
   \SOM{24}{0}{}{}
   \bSOM{25}{0}{}{}
   \SOM{26}{0}{}{}

\tobedone

  \node[above] at (1,1) {$v_{i}$};
   \node[above] at (13,1) {$v_{i+a}$};
   \node[above] at (25,1) {$v_{i+2a}$};

   \node[below] at (6,-0.9) {($a-1$ vertices)};
   \node[below] at (17,-0.9) {($a-1$ vertices)};

  \node[below] at (12,-2.2){(a) two consecutive sequences of $a-1$ vertices };

\end{tikzpicture}

 \caption{\label{fig:} Construction of the mapping $f$ in the proof of  Proposition~\ref{prop:bound4} ($a=12$) .}
\end{center}
\end{figure}

\begin{proposition}\label{prop:bound4}
If $n$, $a$, $k_1$ and $k_2$  are four integers such that $n=k_1 (a+1)+ k_2 (a-1)$, $6\le a\le\left\lfloor\frac{n}{2}\right\rfloor$, and $a$ is even,
then, for every independent broadcast $f$ on $C(n;1,a)$, we have
$$\sigma(f)\geq k_1 \left(\dfrac{a}{2} \right)+k_2 \left(\dfrac{a}{2}-1\right).$$
\end{proposition}

\begin{proof}
Let $n=k_1 (a+1)+ k_2 (a-1)$. The circulant graph $C(n;1,a)$ consists of $k_1$ sequences of $a+1$ vertices and $k_{2}$ sequences of $a-1$ vertices. Let $f$ be a mapping from $V(C(n;1,a)$ to $\{0,1\}$, defined as follows (see Figure~\ref{fig:}  for the case $a=12$). For every sequence of $a+1$ or $a-1$ vertices, we let the broadcast values of the form $1010\ldots10100$. Since $a$ is even, for every two consecutive sequences, the $f$-broadcast vertices are pairwise non adjacent and then,  $f$ is  an independent broadcast on $C(n;1,a)$, with cost $\sigma(f)=k_1 \left(\dfrac{a}{2} \right)+k_2 \left(\dfrac{a}{2}-1\right)$. Hence,
$$\sigma(f)\geq k_1 \left(\dfrac{a}{2} \right)+k_2 \left(\dfrac{a}{2}-1\right).$$
This completes the proof.
\end{proof}

\section{Some exact values}\label{sec:exact}

We determine in this section the broadcast independence number of circulant graphs of the form $C(n;1,a)$, for various values of $n$ and $a$.
In several cases, we prove, thanks to Observation~\ref{obs:beta=alpha}, that the independence number and the broadcast independence number of these graphs coincide.

In \cite{LZY09}, Liancheng, Zunquan and Yuansheng determined the exact value of the independence number of some circulant graphs of the form $C(n;1,a)$.

\begin{proposition}[Liancheng {\it et al.}~\cite{LZY09}]\label{prop:independent}
For every two integers $n$ and $a$ with $2\le a\le\left\lfloor\frac{n}{2}\right\rfloor$, we have
\begin{enumerate}
\item $\alpha(C(n;1,a)) = \frac{n}{2}, \mbox { for even } n \mbox { and odd } a$,
\item $\alpha(C(n;1,a)) = \frac{n-k}{2}, \mbox{ for odd } n \mbox { and }  a\in \{3,5\}$,
\item $\alpha(C(n;1,2)) = \left \lfloor\frac{n}{3}\right \rfloor$,
\item $\alpha(C(n;1,4)) = \left \lfloor\frac{2n}{5} \right \rfloor$.
\end{enumerate}
\end{proposition}

Several of our results in this section will thus extend the results of Proposition~\ref{prop:independent}.



\medskip

We first consider the case of circulant graphs of the form $C(n;1,2)$, $n\ge 4$. It is not difficult to check that, for every $n\ge 4$, antipodal vertices
in $C(n;1,2)$ are at distance $\left\lceil\frac{n-1}{4}\right\rceil$ apart from each other.
We thus have the following.

\begin{observation}\label{obs:diam-C(n;1,2)}
For every integer $n$,  $n\geq 4$,  $\diam(C(n;1,2)) = \left\lceil\dfrac{n-1}{4}\right\rceil$.
\end{observation}

%

The broadcast independence number of circulant graphs of the form $C(n;1,2)$ is given by the following result.

\begin{theorem}\label{th:C(n;1,2)}  
For every integer $n\ge 4$,
$$\beta_b(C(n;1,2)) = \left\{
  \begin{array}{ll}
    \alpha(C(n;1,2)) = 1, & \mbox{if $n\in \{4,5\}$,} \\ [1ex]
    \dfrac{n-3}{2}, & \mbox{if $n\equiv 9 \pmod{12}$,} \\ [1ex]
    2(\diam(C(n;1,2)) - 1) = 2\left(\left\lceil\dfrac{n-1}{4}\right\rceil-1\right), & \mbox{otherwise.}
  \end{array}
\right.$$
\end{theorem}

\begin{proof}
Since $C(4;1,2)$ and $C(5;1,2)$ are both complete graphs, the result obviously holds
for $n\in \{4,5\}$.

Suppose now $n\ge 6$.
By Proposition~\ref{prop:lower bound}, $\beta_b(C(n;1,2)) \geq 2(\diam(C(n;1,2))-1)$ holds
for every $n$.
We will prove that we have $\beta_b(C(n;1,2)) \le 2(\diam(C(n;1,2))-1)$
if $n\not\equiv 9\pmod{12}$,
and $\beta_b(C(n;1,2)) = \frac{n-3}{2}$ otherwise.

Let $f$ be an independent $\beta_b$-broadcast on $C(n;1,2)$.
Each vertex $v\in V_f^+$ $f$-dominates $4f(v)+1$ vertices.
Moreover, each $f$-broadcast vertex is $f$-dominated exactly once, and each
non-broadcast vertex is $f$-dominated at most twice. This gives
$$4f(V_f^+) + |V_f^+| \leq 2\left(n - |V_f^+|\right) + |V_f^+|,$$
and thus
$$\beta_b(C(n;1,2)) = \sigma(f) = \sum_{v\in V_f^+}f(v) = f(V_f^+) \leq \frac{n - |V_f^+|}{2}.$$

We now consider three cases, depending on the value of $|V_f^+|$.
\begin{enumerate}
\item $|V_f^+| \le 2$.\\
If $|V_f^+| = 1$, then $V_f^+=\{v_i\}$ for some vertex $v_i$,
and thus
$$\sigma(f) = f(v_i) \le e(v_i) = \diam(C(n;1,2))\leq 2(\diam(C(n;1,2)) - 1).$$
If $|V_f^+| = 2$, then $V_f^+=\{v_i,v_j\}$ for some distinct vertices $v_i$ and $v_j$,
and thus
$$\sigma(f) = f(v_i) + f(v_j) \le 2(\diam(C(n;1,2)) - 1).$$

\item $|V_f^+| \ge 4$.\\
In that case, we get
$$\sigma(f)
   \le \left\lfloor\frac{n - |V_f^+|}{2}\right\rfloor
   \le \left\lfloor\frac{n-4}{2}\right\rfloor
   \le 2\left(\left\lceil\frac{n-1}{4}\right\rceil - 1\right)$$
and thus $\sigma(f) \le 2\left(\left\lceil\frac{n-1}{4}\right\rceil - 1\right) = 2(\diam(C(n;1,2)) - 1)$
by Observation~\ref{obs:diam-C(n;1,2)}.

\item $|V_f^+| = 3$.\\
Let $V_f^+ = \{v_{i_0},v_{i_1},v_{i_2}\}$, with $0\le i_0 < i_1 < i_2 < n-1$.
We consider two subcases, depending in the parity of $n$.

\begin{enumerate}
\item $n$ is even.\\
Since $f$ is a $\beta_b$-broadcast, we have
$$f(v_{i_j}) = \min\left\{d(v_{i_j},v_{i_{j-1}})-1,d(v_{i_j},v_{i_{j+1}})-1\right\}$$
for every $j$, $0\le j\le 2$ (subscripts are taken modulo~3).

Moreover, since $\min\{x,y\}\le\frac{x+y}{2}$ for every two integers $x$ and $y$, we get
$$\sigma(f) = f(v_{i_0}) + f(v_{i_1}) + f(v_{i_2}) \le
   d(v_{i_{0}},v_{i_{1}}) + d(v_{i_{1}},v_{i_{2}}) + d(v_{i_{2}},v_{i_{0}})-3.$$

Now, since
$$d(v_{i_{j}},v_{i_{j'}}) = \left\lceil  \frac{|{i_{j}} - {i_{j'}}|}{2} \right\rceil
\le \frac{|{i_{j}} - {i_{j'}}| + 1}{2} $$
for every two distinct vertices $v_{i_{j}}$ and $v_{i_{j'}}$, we get
$$\sigma(f) \le \left\lfloor\frac{|{i_{0}} - {i_{1}}| + |{i_{1}} - {i_{2}}| + |{i_{2}} - {i_{0}}| + 3}{2}\right\rfloor - 3
= \left\lfloor\frac{n-3}{2}\right\rfloor.$$
Finally, since $n$ is even, we get
$$\sigma(f) \le \left\lfloor\frac{n-3}{2}\right\rfloor = \frac{n-4}{2}
    \le 2\left(\left\lceil\frac{n-1}{4}\right\rceil - 1\right)
    = 2(\diam(C(n;1,2)) - 1).$$

\item $n$ is odd.\\
If every non-broadcast vertex is $f$-dominated exactly twice,
then we necessarily have $f(v_{i_0}) = f(v_{i_1}) = f(v_{i_2}) = \ell$
for some value $\ell$.
Moreover, since each vertex $v_{i_j}$, $0\le j\le 2$, $f$-dominates $4f(v_{i_j})+1 = 4\ell+1$
vertices, we get $12\ell + 3=2(n-3)+3$ (each vertex in $V_f^+$ is $f$-dominated only once),
and thus $\ell = \frac{n-3}{6}$.
This implies $n\equiv 3\pmod 6$ and
$\sigma(f)=\frac{n-3}{2}$.

Now,  we have
$$\sigma(f) = \frac{n-3}{2} = 3\ell = 2\left(\left\lceil\frac{n-1}{4}\right\rceil - 1\right)
    = 2(\diam(C(n;1,2)) - 1)$$
if $\ell$ is even, that is $n\equiv 3\pmod {12}$, while we have
$$\sigma(f) = \frac{n-3}{2} = 3\ell > 3\ell - 1 = 2\left(\left\lceil\frac{n-1}{4}\right\rceil - 1\right)
    = 2(\diam(C(n;1,2)) - 1)$$
if $\ell$ is odd, that is $n\equiv 9\pmod {12}$.

\medskip
Suppose now that at least one non-broadcast vertex is $f$-dominated only once,
which implies $4f(V_f^+) + 3 \leq 2(n-4)+ 4 = 2n-4$ and thus
$$\sigma(f) =  f(V_f^+)\leq \left\lfloor\dfrac{2n -7}{4} \right\rfloor.$$
Since $n$ is odd, we get
$$\sigma(f) \leq \left\lfloor\dfrac{2n -7}{4} \right\rfloor
  = \left\lfloor\dfrac{2n -8}{4} \right\rfloor
  \leq 2\left(\left\lceil\dfrac{n-1}{4}\right\rceil - 1\right)
  = 2(\diam(C(n;1,2)) - 1).$$

\end{enumerate}
\end{enumerate}

In all cases, we thus get $\beta_b(C(n;1,2)) = \sigma(f) \le f(V_f^+) \le 2(\diam(C(n;1,2)) - 1)$
if $n \not\equiv 9\pmod{12}$,
and $\beta_b(C(n;1,2)) = \sigma(f) = \frac{n-3}{2}$ if $n \equiv 9\pmod{12}$, which completes the proof.
\end{proof}
Comparing the value of $\alpha(C(n;1,2))$ given in \cite{LZY09} with  $\beta_b(C(n;1,2))$, it is clearly seen that  $\alpha(C(n;1,2))<\beta_b(C(n;1,2))$  is almost always true.

Since the circulant graphs $C(2a+1;1,a)$ and $C(2a+1;1,2)$ are isomorphic for every integer $a$, $a\ge 2$, Theorem~\ref{th:C(n;1,2)} admits the following corollary.

\begin{corollary}\label{cor:(C2a+1;1,a)}
For every integer $a\ge 2$,
$$\beta_b(C(2a+1;1,a))=\left\{
  \begin{array}{ll}
  a-1, & \text{if } a=2\text{, or } a\equiv 4\pmod 6,\\ [1ex]
  2\left(\left\lceil\dfrac{a}{2}\right\rceil-1\right), & \text{otherwise.}
\end{array}
\right.
$$
\end{corollary}


We now determine the broadcast independence number of circulant graphs
of the form $C(n;1,a)$ when $n$ is even and $a$ is odd.

\begin{theorem}\label{th:n-even a-odd}
If $n$ and $a$ are two integers such that $n$ is even, $n\ge 6$, $a$ is odd
and $3\le a\le \left\lfloor\frac{n}{2}\right\rfloor$,
then
$$\beta_b(C(n;1,a)) = \alpha(C(n;1,a)) = \frac{n}{2}.$$
\end{theorem}

\begin{proof}
From Proposition~\ref{prop:bound1}, we get that
$\sigma(g)\leq \Big\lfloor\frac{n-\left|V_g^{2}\right|}{2}\Big\rfloor$
for every $2$-bounded independent broadcast $g$ on $C(n;1,a)$,
which implies $\beta_b(C(n;1,a)) \le \frac{n}{2}$.
Consider now the mapping $f$  from $V(C(n;1,a))$ to $\{0,1\}$ defined by $f(v_i)=1$ if
and only if $i$ is even.
Since $a$ is odd, $f$ is clearly an independent broadcast on $C(n;1,a)$.
This implies $\beta_b(C(n;1,a)) \ge \sigma(f) = \frac{n}{2}$ and thus,
thanks to Observation~\ref{obs:beta=alpha}, $\beta_b(C(n;1,a)) = \alpha(C(n;1,a)) = \frac{n}{2}$. This completes the proof.
\end{proof}

We are now able to determine the broadcast independence number of circulant graphs
of the form $C(n;1,3)$.

\begin{theorem}\label{th:C(n;1,3)}
For every integer $n\ge 6$,
$$\beta_b(C(n;1,3))=\alpha(C(n;1,3))=\left\{
  \begin{array}{ll}
    \dfrac{n}{2},    & \text{if $n$ is even},\\ [2ex]
    \dfrac{n-3}{2},  & \text{otherwise. }
  \end{array}
\right.$$
\end{theorem}

\begin{proof}
If $n$ is even, the result directly follows from Theorem~\ref{th:n-even a-odd}.

Suppose now that $n$ is odd and consider the mapping $f$ from $V(C(n;1,3))$ to $\{0,1\}$
defined by $f(v_i)=1$ if and only if $i$ is even and $i\le n-5$.
Since all broadcast vertices have an even index not greater than $n-5$ and $3$ is odd,
$f$ is clearly a $1$-bounded independent broadcast on $C(n;1,3)$
with $\sigma(f) = \frac{n-3}{2}$ and $V_f^{2}=\emptyset$.
We thus get $\beta_b(C(n;1,3)) \ge \frac{n-3}{2}$ and, thanks to Observation~\ref{obs:beta=alpha},
$\beta_b(C(n;1,3)) = \alpha(C(n;1,3))$.

From Proposition~\ref{prop:bound1}, we get that
$\sigma(g)\leq \Big\lfloor\frac{n-\left|V_g^{2}\right|}{2}\Big\rfloor$
for every $2$-bounded independent broadcast $g$ on $C(n;1,3)$.
If $\left|V_g^{2}\right| \ge 2$,
then $\sigma(g)\leq \left\lfloor\frac{n-2}{2}\right\rfloor
= \frac{n-3}{2} = \sigma(f)$.
If $\left|V_g^{2}\right| = 1$, say $V_g^{2}=\{v_j\}$,
then we necessarily have $g(v_{j-1})=g(v_{j-2})=g(v_{j-3})=g(v_{j-4})=0$,
and thus $v_{j-1}$ and $v_{j-2}$ do not belong to any set $A_f^i$ for
any $v_i\in V_g^1$.
Using this remark together with Lemma~\ref{lem:AiBj}, we then get
$$\sum_{v_i\in V_g^1}|A_g^i|  +  \sum_{v_j\in V_g^{2}}|B_g^j|
= 2f(V_g^1) + \left|V_g^1\right| + 3f(v_j) - 1 \le n-2,$$
which gives, since $f(v_j)\ge 2$,
$$\sigma(g) = f(V_g^1) + f(v_j) \le
\frac{n-2-\left|V_g^1\right|-f(v_j)+1}{2}
\le \frac{n-3}{2} = \sigma(f).$$

Finally, if $\left|V_g^{2}\right| = 0$ then, since $n$ is odd,
there necessarily exists a vertex $v_i\in V_g^1$ such that $g(v_{i+2})=0$,
which implies $g(v_{i+1})=g(v_{i+3})=0$.
This implies that
$v_{i+3}$ does not belong to any set $A_g^{i'}$ for any $v_{i'}\in V_g^1$.
Using this remark together with Lemma~\ref{lem:AiBj},
we then have $2f(V_g^1) + \left|V_g^1\right| \le n-1$,
and thus
$$\sigma(g) = f(V_g^1) \le
\frac{n-1-\left|V_g^1\right|}{2}
\le \left\lfloor\frac{n-2}{2}\right\rfloor = \frac{n-3}{2} = \sigma(f).$$

Hence, in all the previous cases, we have $\sigma(g)\le\sigma(f)=\frac{n-3}{2}$,
which completes the proof.
\end{proof}

We now determine the broadcast independence number of circulant graphs
of the form $C(n;1,4)$.

\begin{theorem}\label{th:C(n;1,4)}
For every integer $n\ge 8$,
$$\beta_b(C(n;1,4)) = \alpha(C(n;1,4))= \left \lfloor\dfrac{2n}{5}  \right\rfloor.$$
\end{theorem}

\begin{proof}
From Proposition~\ref{prop:bound2}, we get that
$$\sigma(g)\leq \left\lfloor\frac{a}{2(a+1)}\left(n-\frac{a-4}{a}\left|V_g^{2}\right|\right) \right\rfloor$$
for every $2$-bounded independent broadcast $g$ on $C(n;1,a)$,
which gives $\sigma(g)\leq \left \lfloor\frac{2n}{5}  \right \rfloor$,
and thus $\beta_b(C(n;1,4)) \le \left \lfloor\frac{2n}{5}  \right \rfloor$.

We now construct a mapping $f$ from $V(C(n;1,4))$ to $\{0,1\}$.
Let $n=5k+r$ with $0\le r\le 4$.
We consider five cases, depending on the value of $r$.

\begin{enumerate}
\item $r=0$.\\
We let $f(v_i)=1$ if $(i\mod 5)$ is odd, and $f(v_i)=0$ otherwise.

\item $r=1$. \\
We let $f(v_i)=1$ if $(i\mod 5)$ is odd and $i\leq n-7$,  $f(v_{n-2})=f(v_{n-5})=1$, and $f(v_i)=0$ otherwise.

\item $r=2$. \\
We let $f(v_i)=1$ if $(i\mod 5)$ is odd and $i\leq n-3$, and $f(v_i)=0$ otherwise.

\item $r=3$. \\
We let $f(v_i)=1$ if $(i\mod 5)$ is odd, and $f(v_i)=0$ otherwise.

\item $r=4$. \\
We let $f(v_i)=1$ if $(i\mod 5)$ is odd and $i\leq n-7$,
$f(v_{n-2})=f(v_{n-5})=f(v_{n-5})=1$,  and $f(v_i)=0$ otherwise.
\end{enumerate}

Clearly, in each of the previous cases, $f$ is a $1$-bounded independent broadcast on $C(n;1,4)$
such that $\sigma(f) = \left\lfloor \frac{2n}{5} \right \rfloor $ and $V_f^{2}=\emptyset$.
Hence, $\beta_b(C(n;1,4)) = \left\lfloor \frac{2n}{5} \right\rfloor$ and,
thanks to Observation~\ref{obs:beta=alpha}, $\beta_b(C(n;1,4))=\alpha(C(n;1,4))$. This completes the proof.
\end{proof}

Thanks to Proposition~\ref{prop:bound2},
we are now able to determine the broadcast independence number of circulant graphs
of the form $C((a+1)k;1,a)$ with $a\ge 5$ and $k\geq 2$
(the cases $a=2$, $3$ and $4$ are already covered by Theorems \ref{th:C(n;1,2)},
\ref{th:C(n;1,3)} and~\ref{th:C(n;1,4)}, respectively).

\begin{theorem}\label{th:C((a+1)k;1,a)}
If $a$ and $k$ are two integers such that $a\ge 5$ and $k\geq 2$, then we have
$$\beta_b(C((a+1)k;1,a)) = \alpha(C((a+1)k;1,a)) =
\left\{ \begin{array}{ll}
   \dfrac{ak}{2}, & \text{if $a$ is even},\\[2ex]
   \dfrac{(a+1)k}{2},       & \text{otherwise.}
   \end{array}
\right.
$$
\end{theorem}

\begin{proof}
If $a$ is odd, then $(a+1)k$ is even and the result directly follows from Theorem~\ref{th:n-even a-odd}.

Suppose now that $a$ is even, which implies $a\ge 6$.
From Proposition~\ref{prop:bound2}, we get that
$$\sigma(g)\leq \left\lfloor\frac{a}{2(a+1)}\left((a+1)k-\frac{a-4}{a}\left|V_g^{2}\right|\right) \right\rfloor \le
\left\lfloor\frac{a}{2(a+1)}(a+1)k \right\rfloor = \frac{ak}{2}$$
for every $2$-bounded independent broadcast $g$ on $C((a+1)k;1,a)$,
which implies $\beta_b(C((a+1)k;1,a)) \le \frac{ak}{2}.$

Consider now the mapping $f$  from $V(C((a+1)k;1,a))$ to $\{0,1\}$
defined by $f(v_i)=1$ if and only if $(i \mod a+1)$ is odd.
Since $a$ is even, $f$ is clearly a $1$-bounded independent broadcast on $C((a+1)k;1,a)$
with $\sigma(f) = \frac{ak}{2}$ and $V_f^{2}=\emptyset$.
This implies $\beta_b(C((a+1)k;1,a)) \ge \frac{ak}{2}$ and thus,
thanks to Observation~\ref{obs:beta=alpha},
$\beta_b(C((a+1)k;1,a)) = \alpha(C((a+1)k;1,a)) = \frac{ak}{2}$. This completes the proof.
\end{proof}

%

We now consider the case of circulant graphs $C(n;1,a)$ when $a$ divides $n$.
%
We first introduce two new sets of vertices,
slightly modifying the definition of the sets $A_f^i$ and $B_f^j$
defined in Section~\ref{sec:preliminaries}, using $a$-edges instead of $1$-edges.
Let $f$ be any $2$-bounded independent broadcast on $C(n;1,a)$.
Now consider any vertex $v_i\in V_f^1$ such that $f(v_{i-a}) = f(v_{i-2a}) = 0$.
Since $f$ is an independent broadcast, we necessarily have $f(v_{i+a}) = 0$.
Moreover, we then have either $f(v_{i+2a}) = 0$ or $f(v_{i+2a}) = 1$.
Therefore, the broadcast values of the sequence of vertices $v_iv_{i+a}v_{i+2a}\dots$
is of the form either $100$, $10100$ or $1010\dots100$.

For each vertex $v_i\in V_f^1$ such that $f(v_{i-a}) = f(v_{i-2a}) = 0$,
we then let
$${A'}_f^i = \{v_{i+\ell a},\ 0\le \ell\le 2p+2\}$$
be the set of vertices satisfying
(i) $f(v_{i+2ka})=1$ and $f(v_{i+(2k+1)a})=0$ for every $k$, $0\le k\le p$,
and (ii) $f(v_{i+(2p+2)a})=0$.

Now, for each vertex $v_j\in V_f^{2}$, we let
$${B'}_f^j =  \{v_j\}\cup\{v_{j+a-1},v_{j+a},v_{j+a+1}\}\cup\{v_{j+2a}\}.$$

These sets satisfy the same properties as those of the sets
$A_f^i$ and $B_f^j$ given in Lemma~\ref{lem:AiBj}.
The proof is similar to the proof of Lemma~\ref{lem:AiBj} and is omitted.

\begin{lemma}\label{lem:AiBj-Prime}
For every $2$-bounded independent broadcast $f$ on $C(n;1,a)$, if any, the following holds.
\begin{enumerate}
\item For every vertex $v_i\in V_f^1$, $|{A'}_f^i| = 2f({A'}_f^i)+1$.
\item For every vertex $v_j\in V_f^{2}$, $|{B'}_f^j| = 5$.
\item $\sum_{v_i\in V_f^1}|{A'}_f^i|  +  \sum_{v_j\in V_f^{2}}|{B'}_f^j| \le n$.
\end{enumerate}
\end{lemma}

We are now ready to determine the independent broadcast number of circulant graphs of the form $C(qa;1,a)$, $a\ge 5$ and $q\ge 4$.
Recall that the cases $a=2$, $3$ and $4$ are already covered by Theorems \ref{th:C(n;1,2)}, \ref{th:C(n;1,3)} and~\ref{th:C(n;1,4)}, respectively,
while the cases $q=2$ and $q=3$ are covered by Theorems \ref{th:C(2a;1,a)} and~\ref{th:C(3a;1,a)}, respectively.

\begin{theorem}\label{th:C(qa;1,a)}
If $a$ and $q$ are two integers such that $a\ge 5$ and $q\ge 4$, then we have
$$\begin{array}{ll}
\beta_b(C(qa;1,a)) & =\ \alpha(C(qa;1,a)) \\ [2ex]
 & =\
\left\{
  \begin{array}{ll}
    \dfrac{qa}{2},     & \mbox{if $a$ is odd, and $q$ is even,} \\ [2ex]
    \dfrac{(q-1)a}{2}, & \mbox{if $a$ and $q$ are odd,} \\ [2ex]
    \left\lfloor \dfrac{qa^2}{2(a+1)} \right\rfloor, & \mbox{if $a$ and $q$ are even,} \\ [2ex]
    \min\left\{ \left\lfloor \dfrac{qa^2}{2(a+1)} \right\rfloor, \dfrac{(q-1)a}{2}\right\},
              & \mbox{otherwise.}
  \end{array}
\right.
\end{array}$$
\end{theorem}

\begin{proof}
We consider the four cases separately.
\begin{enumerate}
\item $a$ is odd and $q$ is even.\\
In that case, $qa$ is even and the result directly follows from Theorem~\ref{th:n-even a-odd}.

\item $a$ and $q$ are odd.\\
In that case, $q\ge 5$ and we know by Lemma~\ref{lem:broadcast-at-most-2-general} that $C(qa;1,a)$ admits a $2$-bounded $\beta_b$-broadcast.
Let $f$ be any $2$-bounded independent broadcast on $C(qa;1,a)$.
Observe first that, since $q$ is odd, we necessarily have $|{A'}_f^i|\le q$ for every vertex $v_i\in V_f^1$,
since otherwise this would give $f(v_i)=f(v_{i+(q-1)a})=f(v_{i-a})=1$, contradicting the fact that $f$ is an independent broadcast.
Therefore  $f({A'}_f^i)\le \frac{q-1}{2}$, and  thanks to Item 1 of Lemma~\ref{lem:AiBj-Prime}, we get
$$\frac{|{A'}_f^i|}{f({A'}_f^i)} = \frac{2f({A'}_f^i)+1}{f({A'}_f^i)} = 2 + \frac{1}{f({A'}_f^i)}
\ge 2 + \frac{2}{q-1} = \frac{2q}{q-1},$$
which gives
$$|{A'}_f^i| \ge \frac{2q}{q-1}f({A'}_f^i).$$

Now, using Item 2 and Item 3 of Lemma~\ref{lem:AiBj-Prime}, we get
$$qa \ge \sum_{v_i\in V_f^1}|{A'}_f^i|  +  \sum_{v_j\in V_f^{2}}|{B'}_f^j|
   \ge  \frac{2q}{q-1}f(V_f^1) + \dfrac{5}{2}f(V_f^{2}) ,$$
which gives
$$qa \ge  \frac{2q}{q-1}\sigma(f)  + \dfrac{q-5}{2(q-1)}f(V_f^{2})\geq \frac{2q}{q-1}\sigma(f),$$
and then
$$\sigma(f) \le \dfrac{(q-1)a}{2}.$$

Consider now the mapping $g$ from $V(C(qa;1,a))$ to $\{0,1\}$
defined by  $g(v_i)=1$ if $i$ is even and $i\leq (q-1)a-1$, and $g(v_i)=0$ otherwise.
Clearly, $g$ is a $1$-bounded independent broadcast on $C(qa;1,a)$.
Moreover,
$$\beta_b(C(qa;1,a)) \ge \sigma(g) =  \frac{(q-1)a}{2},$$ and thus,
thanks to Observation~\ref{obs:beta=alpha},
$$\beta_b(C(qa;1,a)) = \alpha(C(qa;1,a)) = \frac{(q-1)a}{2}.$$

\item $a$ and $q$ are even.\\
Note first that if $a+1$ divides $q$, say $q=\ell(a+1)$ for some integer $\ell\ge 1$,
which gives $qa=\ell a(a+1)$,
the result directly follows from Theorem~\ref{th:C((a+1)k;1,a)}
for $k=\ell a$, since
$$\dfrac{ak}{2} = \dfrac{\ell a^2}{2} = \dfrac{qa^2}{2(a+1)} =
\left\lfloor \dfrac{qa^2}{2(a+1)} \right\rfloor.$$

Assume now that this is not the case, so that $k(a+1) < q < (k+2)(a+1)$ for some even integer $k\geq 2$.
Let $q= k(a+1) + 2\ell$ (recall that $k$ and $q$ are even) for some integer $\ell$, $1\le\ell\le a$.  From Proposition~\ref{prop:bound2}, we get that
$$\sigma(f) \leq \left\lfloor\frac{a}{2(a+1)}\left(qa-\frac{a-4}{a}\left|V_f^{2}\right|\right) \right\rfloor \le \left\lfloor\frac{qa^2}{2(a+1)} \right\rfloor$$
for every $2$-bounded independent broadcast $f$ on $C(qa;1,a)$.  Moreover, we have
$$ \left\lfloor\frac{qa^2}{2(a+1)} \right\rfloor  = \dfrac{(a-1)q}{2}
+ \left\lfloor\frac{q}{2(a+1)} \right\rfloor  =  \dfrac{(a-1)q}{2} + \dfrac{k}{2},  $$
which implies
$$\beta_b(C(qa;1,a)) \le  \dfrac{(a-1)q}{2} + \dfrac{k}{2}.$$

Since $qa=(ka+ \ell )(a+1) + \ell(a-1) $, and thanks to Proposition~\ref{prop:bound4}, we get
$$\beta_b(C(qa;1,a)) \geq  \frac{(ka+ \ell)a}{2}+ \frac{\ell(a-2) }{2}  =\frac{ka^2}{2} + \ell(a-1)= \frac{ka^2}{2} + \dfrac{q-k(a+1)}{2}(a-1),$$
which gives
$$\beta_b(C(qa;1,a)) \geq \dfrac{(a-1)q}{2} + \dfrac{k}{2}.$$

Finally, by Observation~\ref{obs:beta=alpha}, we get
$$\beta_b(C(qa;1,a)) = \alpha(C(qa;1,a)) = \left\lfloor\frac{qa^2}{2(a+1)} \right\rfloor.$$

\begin{figure}
\begin{center}
\begin{tikzpicture}[scale=0.7]

  \LIGNE {0}{0}{1}{-0.7}
  \LIGNE {1}{-0.7}{1.5}{-1.5}
  \LIGNE {1.5}{-1.5}{1.7}{-2.5}
  \LIGNE {1.7}{-2.5}{1.6}{-3.5}
  \LIGNE {-2.5}{-1.5}{-2}{-0.7}
  \LIGNE {-2}{-0.7}{-1}{0}
  \LIGNE {-1}{0}{0}{0}

  \LIGNE {1}{0.5}{2}{0.5}
  \LIGNE {2}{0.5}{3}{-0.2}
  \LIGNE {3}{-0.2}{3.5}{-1}
  \LIGNE {3.5}{-1}{3.7}{-2}
  \LIGNE {3.7}{-2}{3.6}{-3}

   \LIGNE {1}{0.5}{-1}{0}
  \LIGNE {2}{0.5}{0}{0}
  \LIGNE {3}{-0.2}{1}{-0.7}
  \LIGNE {3.5}{-1}{1.5}{-1.5}
  \LIGNE {3.7}{-2}{1.7}{-2.5}
  \LIGNE {3.6}{-3}{1.6}{-3.5}

    \LIGNE {5}{1.5}{6}{1.5}
  \LIGNE {6}{1.5}{7}{0.8}
  \LIGNE {7}{0.8}{7.5}{0}
  \LIGNE {7.5}{0}{7.7}{-1}
  \LIGNE {7.7}{-1}{7.6}{-2}

  \LIGNE {9}{2.5}{10}{2.5}
  \LIGNE {10}{2.5}{11}{1.8}
  \LIGNE {11}{1.8}{11.5}{1}
  \LIGNE {11.5}{1}{11.7}{0}
  \LIGNE {11.7}{0}{11.6}{-1}

  \LIGNE {1.6}{-3.5}{1.5}{-3.9}
  \POINTILLE {1.5}{-3.9}{1.4}{-4.2}
  \LIGNE {-2.5}{-1.5}{-2.6}{-2}
  \POINTILLE {-2.6}{-2}{-2.7}{-2.3}

  \LIGNE {3.6}{-3}{3.5}{-3.4}
  \POINTILLE {3.5}{-3.4}{3.4}{-3.7}

  \LIGNE {2}{0.5}{2.4}{0.6}
  \POINTILLE {2.4}{0.6}{2.8}{0.7}
  \LIGNE {3}{-0.2}{3.4}{-0.1}
  \POINTILLE {3.4}{-0.1}{3.8}{0}
  \LIGNE {1}{0.5}{1.4}{0.6}
  \POINTILLE {1.4}{0.6}{1.8}{0.7}*

  \LIGNE {5}{1.5}{4.6}{1.4}
  \POINTILLE {4.6}{1.4}{4.2}{1.3}
   \LIGNE {7}{0.8}{6.6}{0.7}
  \POINTILLE {6.6}{0.7}{6.2}{0.6}
  \LIGNE {7.5}{0}{7.1}{-0.1}
  \POINTILLE {7.1}{-0.1}{6.7}{-0.2}
\LIGNE {7.7}{-1}{7.3}{-1.1}
  \POINTILLE {7.3}{-1.1}{6.9}{-1.2}
 \LIGNE {7.6}{-2}{7.2}{-2.1}
  \POINTILLE {7.2}{-2.1}{6.8}{-2.2}
  \LIGNE {6}{1.5}{5.6}{1.4}
  \POINTILLE {5.6}{1.4}{5.2}{1.3}

  \LIGNE {10}{2.5}{9.4}{2.4}
  \POINTILLE {9.4}{2.4}{9}{2.3}
   \LIGNE {11}{1.8}{10.6}{1.7}
  \POINTILLE {10.6}{1.7}{10.2}{1.6}
  \LIGNE {11.5}{1}{11.1}{0.9}
  \POINTILLE {11.1}{0.9}{10.7}{0.8}

\LIGNE {11.7}{0}{11.3}{-0.1}
  \POINTILLE {11.3}{-0.1}{10.9}{-0.2}
 \LIGNE {11.6}{-1}{11.2}{-1.1}
  \POINTILLE {11.2}{-1.1}{10.8}{-1.2}
  \LIGNE {9}{2.5}{8.6}{2.4}
  \POINTILLE {8.6}{2.4}{8.2}{2.3}

  \LIGNE {5}{1.5}{5.4}{1.6}
  \POINTILLE {5.4}{1.6}{5.8}{1.7}
   \LIGNE {7}{0.8}{7.4}{0.9}
  \POINTILLE {7.4}{0.9}{7.8}{1}
  \LIGNE {6}{1.5}{6.4}{1.6}
  \POINTILLE {6.4}{1.6}{6.8}{1.7}

 \LIGNE {7.6}{-2}{7.5}{-2.4}
  \POINTILLE {7.5}{-2.4}{7.4}{-2.7}

  \LIGNE {11.6}{-1}{11.5}{-1.4}
  \POINTILLE {11.5}{-1.4}{11.4}{-1.7}

  \POINTILLE {13}{3}{12}{3}
  \POINTILLE {13}{3}{14}{2.3}
  \POINTILLE {14}{2.3}{14.5}{1.5}
  \POINTILLE {14.5}{1.5}{14.7}{0.5}
  \POINTILLE {14.7}{0.5}{14.6}{-0.5}
  \POINTILLE {14.6}{-0.5}{14.5}{-1.1}

  \POINTILLE {10}{2.5}{13}{3}
  \POINTILLE {11}{1.8}{14}{2.3}
  \POINTILLE {11.5}{1}{14.5}{1.5}
  \POINTILLE {11.7}{0}{14.7}{0.5}
  \POINTILLE {11.6}{-1}{14.6}{-0.5}

   \SOM{0}{0}{}{\scriptsize 0}
   \drSOM{1}{-0.7}{}{{\scriptsize a}}
   \drSOM{1.5}{-1.5}{}{\scriptsize 2a}
   \drSOM{1.7}{-2.5}{}{\scriptsize 3a}
   \drSOM{1.6}{-3.5}{}{\scriptsize 4a}
 \ghSOM{-1}{0}{}{\scriptsize (q-1)a}
 \ghSOM{-2}{-0.7}{}{\scriptsize (q-2)a}
 \ghSOM{-2.5}{-1.5}{}{\scriptsize (q-3)a}

  \SOM{2}{0.5}{}{\scriptsize 1}
   \drSOM{3}{-0.2}{}{\quad \scriptsize a+1}
   \drSOM{3.5}{-1}{}{\quad \scriptsize 2a+1}
   \drSOM{3.7}{-2}{}{\quad \scriptsize 3a+1}
   \drSOM{3.6}{-3}{}{\quad \scriptsize 4a+1}
 \SOM{1}{0.5}{\scriptsize (q-1)a+1}{ }

   \SOM{6}{1.5}{}{\scriptsize i}
   \drSOM{7}{0.8}{}{\quad \scriptsize a+i}
   \drSOM{7.5}{0}{}{\quad \scriptsize 2a+i}
   \drSOM{7.7}{-1}{}{ \quad \scriptsize 3a+i}
   \drSOM{7.6}{-2}{}{\quad \scriptsize 4a+i}
   \SOM{5}{1.5}{\scriptsize (q-1)a+i}{}

  \SOM{10}{2.5}{\scriptsize a-1}{}
   \drSOM{11}{1.8}{}{\quad \scriptsize 2a-1}
   \drSOM{11.5}{1}{}{\quad \scriptsize 3a-1}
   \drSOM{11.7}{0}{}{\quad \scriptsize 4a-1}
   \drSOM{11.6}{-1}{}{\quad \scriptsize  5a-1}
 \SOM{9}{2.5}{\scriptsize qa-1}{}

\SOM{12}{3}{\scriptsize 0}{}
\SOM{13}{3}{\scriptsize a}{}
   \drSOM{14}{2.3}{}{\scriptsize 2a}
   \drSOM{14.5}{1.5}{}{\scriptsize 3a}
   \drSOM{14.7}{0.5}{}{\scriptsize 4a}
   \drSOM{14.6}{-0.5}{}{\scriptsize 5a}

   \LIGNE {-2}{-4.8}{-1}{-5.5}
   \LIGNE {-2}{-4.8}{-2.4}{-3.8}
    \POINTILLE {-1}{-5.5}{-0.4}{-5.6}
     \POINTILLE {-2.4}{-3.8}{-2.6}{-3.2}
   \ghSOM{-2}{-4.8}{}{\scriptsize (i+1)a}
    \ghSOM{-1}{-5.5}{}{\scriptsize ia}
   \ghSOM{-2.4}{-3.8}{}{\scriptsize (i+2)a }

 \node[below] at (1.3,-4.5){$C_0$};
 \node[below] at (3.3,-4){$C_1$};
 \node[below] at (7.3,-3){$C_i$};
 \node[below] at (11.4,-1.9){$C_ {(a-1)} $};
  \node[below] at (14.3,-1.5){$C_0$};

\end{tikzpicture}
\caption{\label{fig:Ci} The circulant graph $C(qa;1,a)$ (only subscripts of vertices are indicated)}
\end{center}
\end{figure}
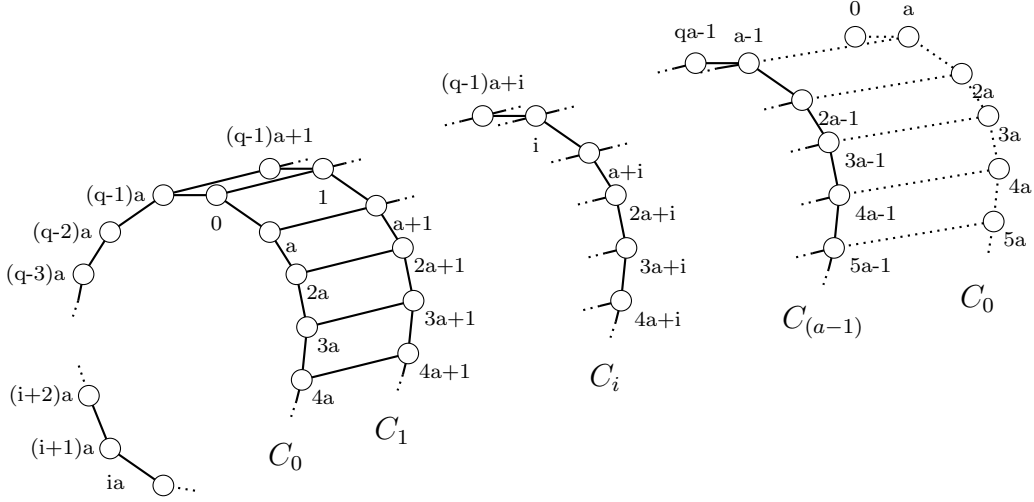

\item $a$ is even and $q$ is odd.\\
The graph $C(qa;1,a)$  can be seen as $a$ copies $C_0,\dots,C_{a-1}$ of a $q$-cycle, with $C_k=\{v_k,v_{k+a},v_{k+2a} \ldots ,v_{k+(q-1)a}\}$, 
for every $k$, $0\le k \leq  a-1$, cyclically connected as depicted in Figure~\ref{fig:Ci}.

%

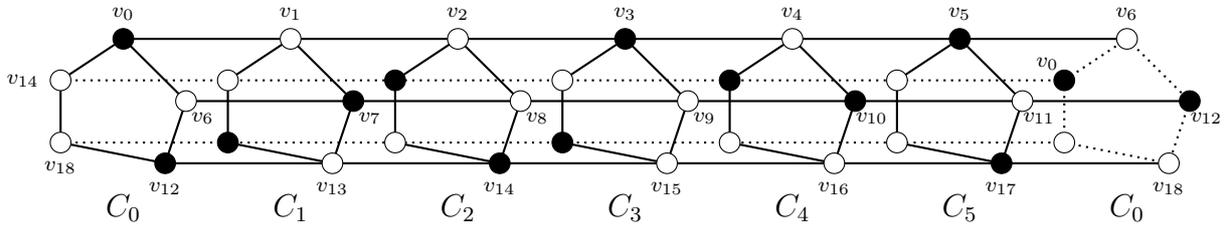
\begin{figure}
\begin{center}
\begin{tikzpicture}[scale=0.55]

  \LIGNE {0}{0}{24}{0}
  \LIGNE {25.5}{-1.5}{1.5}{-1.5}
  \LIGNE {1}{-3}{25}{-3}

  \POINTILLE {-1.5}{-1}{22.5}{-1}
  \POINTILLE {-1.5}{-2.5}{22.5}{-2.5}

   \LIGNE {0}{0}{1.5}{-1.5}
  \LIGNE {1.5}{-1.5}{1}{-3}
  \LIGNE {1}{-3}{-1.5}{-2.5}
  \LIGNE {-1.5}{-2.5}{-1.5}{-1}
  \LIGNE {-1.5}{-1}{0}{0}

     \bSOM{0}{0}{\scriptsize $v_{0}$}{}
      \drSOM{1.5}{-1.5}{}{\scriptsize $v_{6}$}
       \bSOM{1}{-3}{}{\scriptsize $v_{12}$}
        \SOM{-1.5}{-2.5}{}{\scriptsize $v_{18}$}
         \ghSOM{-1.5}{-1}{}{\scriptsize $v_{14}$}

   \LIGNE {4}{0}{5.5}{-1.5}
  \LIGNE {5.5}{-1.5}{5}{-3}
  \LIGNE {5}{-3}{2.5}{-2.5}
  \LIGNE {2.5}{-2.5}{2.5}{-1}
  \LIGNE {2.5}{-1}{4}{0}

     \SOM{4}{0}{\scriptsize $v_{1}$}{}
      \bdrSOM{5.5}{-1.5}{}{\scriptsize $v_{7}$}
       \SOM{5}{-3}{}{\scriptsize $v_{13}$}
        \bSOM{2.5}{-2.5}{}{}
         \SOM{2.5}{-1}{}{}

   \LIGNE {8}{0}{9.5}{-1.5}
  \LIGNE {9.5}{-1.5}{9}{-3}
  \LIGNE {9}{-3}{6.5}{-2.5}
  \LIGNE {6.5}{-2.5}{6.5}{-1}
  \LIGNE {6.5}{-1}{8}{0}

     \SOM{8}{0}{\scriptsize $v_{2}$}{}
      \drSOM{9.5}{-1.5}{}{\scriptsize $v_{8}$}
       \bSOM{9}{-3}{}{\scriptsize $v_{14}$}
        \SOM{6.5}{-2.5}{}{}
         \bSOM{6.5}{-1}{}{}

   \LIGNE {12}{0}{13.5}{-1.5}
  \LIGNE {13.5}{-1.5}{13}{-3}
  \LIGNE {13}{-3}{10.5}{-2.5}
  \LIGNE {10.5}{-2.5}{10.5}{-1}
  \LIGNE {10.5}{-1}{12}{0}

     \bSOM{12}{0}{\scriptsize $v_{3}$}{}
      \drSOM{13.5}{-1.5}{}{\scriptsize $v_{9}$}
       \SOM{13}{-3}{}{\scriptsize $v_{15}$}
        \bSOM{10.5}{-2.5}{}{}
         \SOM{10.5}{-1}{}{}

   \LIGNE {16}{0}{17.5}{-1.5}
  \LIGNE {17.5}{-1.5}{17}{-3}
  \LIGNE {17}{-3}{14.5}{-2.5}
  \LIGNE {14.5}{-2.5}{14.5}{-1}
  \LIGNE {14.5}{-1}{16}{0}

     \SOM{16}{0}{\scriptsize $v_{4}$}{}
      \bdrSOM{17.5}{-1.5}{}{\scriptsize $v_{10}$}
       \SOM{17}{-3}{}{\scriptsize $v_{16}$}
        \SOM{14.5}{-2.5}{}{}
         \bSOM{14.5}{-1}{}{}

   \LIGNE {20}{0}{21.5}{-1.5}
  \LIGNE {21.5}{-1.5}{21}{-3}
  \LIGNE {21}{-3}{18.5}{-2.5}
  \LIGNE {18.5}{-2.5}{18.5}{-1}
  \LIGNE {18.5}{-1}{20}{0}

     \bSOM{20}{0}{\scriptsize $v_{5}$}{}
      \drSOM{21.5}{-1.5}{}{\scriptsize $v_{11}$}
       \bSOM{21}{-3}{}{\scriptsize $v_{17}$}
        \SOM{18.5}{-2.5}{}{}
         \SOM{18.5}{-1}{}{}

 \POINTILLE {24}{0}{25.5}{-1.5}
 \POINTILLE {25.5}{-1.5}{25}{-3}
 \POINTILLE {25}{-3}{22.5}{-2.5}
 \POINTILLE {22.5}{-2.5}{22.5}{-1}
 \POINTILLE {22.5}{-1}{24}{0}

     \SOM{24}{0}{\scriptsize $v_{6}$}{}
      \bdrSOM{25.5}{-1.5}{}{\scriptsize $v_{12}$}
       \SOM{25}{-3}{}{\scriptsize $v_{18}$}
        \SOM{22.5}{-2.5}{}{}
         \bdrSOM{22.5}{-1}{\scriptsize $v_0$}{}

\node[below] at (0,-3.5){$C_0$};
 \node[below] at (4,-3.5){$C_1$};
 \node[below] at (8,-3.5){$C_ {2} $};
 \node[below] at (12,-3.5){$C_ {3}$};
 \node[below] at (16,-3.5){$C_ {4}$};
 \node[below] at (20,-3.5){$C_ {5}$};
\node[below] at (24,-3.5){$C_0$};

\end{tikzpicture}
\caption{\label{fig:Sig}Construction of the sets $S_i$ in the proof of  Theorem~\ref{th:C(qa;1,a)} ($a=6$, $q=5$)}
\end{center}
\end{figure}

We consider three subcases.
\begin{enumerate}

\item $q = a-1$.\\
Since $q$ is odd, similarly to Case~2 ($q$ and $a$ odd), we have
$$\sigma(f) \leq \frac{(q-1)a}{2}$$
for every independent broadcast $f$ on $C(qa;1,a)$.
Consider the sets $S_k$, $0\leq k\leq a-1$ defined as follows (see Figure~\ref{fig:Sig}  for the case $a=6$ and $q=5$).
\begin{align*}
   S_0 &= \{v_0,v_{2a}, v_{4a}, \ldots , v_{(q-5)a}, v_{(q-3)a} \}, \\
   S_1 &= \{v_{1+a},v_{1+3a}, v_{1+5a}, \ldots , v_{1+(q-4)a}, v_{1+(q-2)a} \}, \\
   S_2 & = \{v_{2+2a},v_{2+4a}, v_{2+6a}, \ldots , v_{2+(q-3)a}, v_{2+(q-1)a} \},\\
   &\vdots\\
 S_{a-1}& = \{v_{a-1},v_{3a-1}, v_{5a-1}, \ldots , v_{n-4a-1}, v_{n-2a-1} \}.
\end{align*}

From this definition, we clearly get that $\bigcup_{k=0}^{a-1} S_k$ is a independent set in $V(G)$. This gives
$$\left| \bigcup_{k=0}^{a-1} S_k \right| = a\dfrac{q-1}{2} \leq \alpha(C(qa;1,a))\leq \beta_b(C(qa;1,a)).$$

Thanks to Observation~\ref{obs:beta=alpha}, we then get
 $$\beta_b(C(qa;1,a))= \alpha (C(qa;1,a))=\dfrac{(q-1)a}{2}.$$

\begin{figure}
\begin{center}
\begin{tikzpicture}[scale=0.5]

  \LIGNE {2}{0}{3}{0}
  \POINTILLE {3}{0}{3.5}{0}
  \LIGNE {3.5}{-1.5}{4.5}{-1.5}
  \POINTILLE {4.5}{-1.5}{5}{-1.5}
  \LIGNE {3}{-3}{4}{-3}
  \POINTILLE {4}{-3}{4.5}{-3}

  \LIGNE {7}{0}{28}{0}
  \LIGNE {8.5}{-1.5}{29.5}{-1.5}
  \LIGNE {8}{-3}{29}{-3}
  \POINTILLE {5.5}{-2.5}{26.5}{-2.5}
  \POINTILLE {5.5}{-1}{26.5}{-1}

  \POINTILLE {6.5}{0}{7}{0}
  \POINTILLE {8}{-1.5}{8.5}{-1.5}
  \POINTILLE {7.5}{-3}{8}{-3}

   \LIGNE {2}{0}{3.5}{-1.5}
  \LIGNE {3.5}{-1.5}{3}{-3}
  \LIGNE {3}{-3}{0.5}{-2.5}
  \LIGNE {0.5}{-2.5}{0.5}{-1}
  \LIGNE {0.5}{-1}{2}{0}

     \bSOM{2}{0}{\scriptsize $v_{0}$}{}
      \drSOM{3.5}{-1.5}{}{\scriptsize $v_{10}$}
       \bSOM{3}{-3}{}{\scriptsize $v_{20}$}
        \SOM{0.5}{-2.5}{}{\scriptsize $v_{30}$}
         \ghSOM{0.5}{-1}{}{\scriptsize $v_{40}$}

   \LIGNE {8}{0}{9.5}{-1.5}
  \LIGNE {9.5}{-1.5}{9}{-3}
  \LIGNE {9}{-3}{6.5}{-2.5}
  \LIGNE {6.5}{-2.5}{6.5}{-1}
  \LIGNE {6.5}{-1}{8}{0}

     \bSOM{8}{0}{\scriptsize $v_{5}$}{}
      \drSOM{9.5}{-1.5}{}{\scriptsize $v_{15}$}
       \bSOM{9}{-3}{}{\scriptsize $v_{25}$}
        \SOM{6.5}{-2.5}{}{\scriptsize $v_{35}$}
         \drSOM{6.5}{-1}{}{\scriptsize $v_{45}$}

   \LIGNE {12}{0}{13.5}{-1.5}
  \LIGNE {13.5}{-1.5}{13}{-3}
  \LIGNE {13}{-3}{10.5}{-2.5}
  \LIGNE {10.5}{-2.5}{10.5}{-1}
  \LIGNE {10.5}{-1}{12}{0}

     \SOM{12}{0}{\scriptsize $v_{6}$}{}
      \bdrSOM{13.5}{-1.5}{}{ \scriptsize $v_{16}$}
       \SOM{13}{-3}{}{ \scriptsize $v_{26}$}
        \bSOM{10.5}{-2.5}{}{}
         \SOM{10.5}{-1}{}{}

   \LIGNE {16}{0}{17.5}{-1.5}
  \LIGNE {17.5}{-1.5}{17}{-3}
  \LIGNE {17}{-3}{14.5}{-2.5}
  \LIGNE {14.5}{-2.5}{14.5}{-1}
  \LIGNE {14.5}{-1}{16}{0}

     \bSOM{16}{0}{ \scriptsize $v_{7}$}{}
      \drSOM{17.5}{-1.5}{}{ \scriptsize $v_{17}$}
       \bSOM{17}{-3}{}{ \scriptsize $v_{27}$}
        \SOM{14.5}{-2.5}{}{}
         \SOM{14.5}{-1}{}{}

   \LIGNE {20}{0}{21.5}{-1.5}
  \LIGNE {21.5}{-1.5}{21}{-3}
  \LIGNE {21}{-3}{18.5}{-2.5}
  \LIGNE {18.5}{-2.5}{18.5}{-1}
  \LIGNE {18.5}{-1}{20}{0}

     \SOM{20}{0}{ \scriptsize $v_{8}$}{}
      \bdrSOM{21.5}{-1.5}{}{\scriptsize $v_{18}$}
       \SOM{21}{-3}{}{\scriptsize $v_{28}$}
        \bSOM{18.5}{-2.5}{}{}
         \SOM{18.5}{-1}{}{}

\LIGNE {24}{0}{25.5}{-1.5}
  \LIGNE {25.5}{-1.5}{25}{-3}
  \LIGNE {25}{-3}{22.5}{-2.5}
  \LIGNE {22.5}{-2.5}{22.5}{-1}
  \LIGNE {22.5}{-1}{24}{0}

     \bSOM{24}{0}{\scriptsize $v_{9}$}{}
      \drSOM{25.5}{-1.5}{}{\scriptsize $v_{19}$}
       \bSOM{25}{-3}{}{\scriptsize $v_{29}$}
        \SOM{22.5}{-2.5}{}{}
         \SOM{22.5}{-1}{}{}

 \LIGNE {28}{0}{29.5}{-1.5}
  \LIGNE {29.5}{-1.5}{29}{-3}
  \LIGNE {29}{-3}{26.5}{-2.5}
  \LIGNE {26.5}{-2.5}{26.5}{-1}
  \LIGNE {26.5}{-1}{28}{0}

     \SOM{28}{0}{\scriptsize $v_{10}$}{}
      \bdrSOM{29.5}{-1.5}{}{\scriptsize $v_{20}$}
       \SOM{29}{-3}{}{\scriptsize $v_{30}$}
        \bSOM{26.5}{-2.5}{}{}
         \SOM{26.5}{-1}{}{}

\node[below] at (2,-3.5){$C_0$};
 \node[below] at (8,-3.5){$C_ {q} $};
 \node[below] at (12,-3.5){$C_ {q+1}$};
 \node[below] at (16,-3.5){$C_ {q+2}$};
 \node[below] at (20,-3.5){$C_ {q+3}$};
\node[below] at (24,-3.5){$C_ {q+4}$};
\node[below] at (28,-3.5){$C_0$};

   \LIGNE{1}{-4.5}{1}{-5}
   \LIGNE{1}{-5}{3}{-5}
   \POINTILLE {3}{-5}{4}{-5}
  \POINTILLE {6}{-5}{7}{-5}
   \LIGNE{7}{-5}{9}{-5}
   \LIGNE{9}{-5}{9}{-4.5}

\node[below] at (5,-5.3){$q+1$ cycles};

 \LIGNE{12}{-4.5}{12}{-5}
   \LIGNE{12}{-5}{14}{-5}
   \POINTILLE {14}{-5}{15}{-5}
  \POINTILLE {21}{-5}{22}{-5}
   \LIGNE{22}{-5}{24}{-5}
   \LIGNE{24}{-5}{24}{-4.5}

\node[below] at (18,-5.3){$2\ell$ cycles};

\end{tikzpicture}

\caption{\label{fig:Sil}
Construction of the sets $S_i$ in the proof of  Theorem~\ref{th:C(qa;1,a)} ($a=10$, $q=5$, $\ell =2$)}
\end{center}
\end{figure}
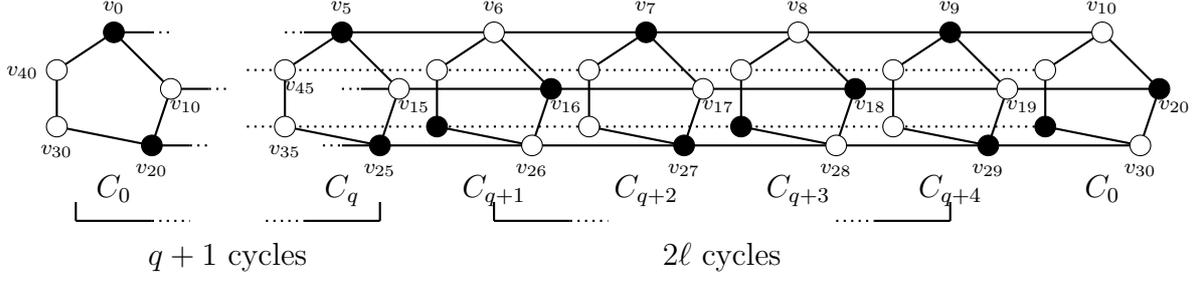

\item $q < a-1$.\\
Let $a=q+1+2\ell $ (recall that $a$ and $q$ have different parity) for some integer $\ell$, $\ell \geq 1 $.
For every cycle $C_k$, $k = 0,1,\ldots,q$, let $S_k$ be the set defined as in the previous subcase and, for every  cycle $C_k$, $k=q+1,\ldots, q+2\ell$, let
 (see Figure~\ref{fig:Sil} for the case $a=10$, $q=5$ and $\ell=2$)
\begin{align*}
   S_k &= \{v_k,v_{k+2a}, v_{k+4a}, \ldots , v_{k+(q-5)a}, v_{k+(q-3)a} \}, \mbox { if $k$ is even}, \\
 S_k &= \{v_{k+a},v_{k+3a}, v_{k+5a}, \ldots , v_{k+(q-4)a}, v_{k+(q-2)a}\}, \mbox { if $k$ is odd}.
 \end{align*}

From this definition, we clearly get that  $\bigcup_{k=0}^{a-1} S_k$ is an independent set of $C(qa;1,a)$. We then have
$$\vert \bigcup_{k=0}^{a-1} S_k \vert = a\dfrac{q-1}{2} \leq \alpha(C(qa;1,a))\leq \beta_b(C(qa;1,a))$$
and, since $q$ is odd, we get
$$\beta_b (C(qa;1,a)\leq a\left ( \dfrac{q-1}{2}\right ).$$

Thanks to Observation~\ref{obs:beta=alpha}, we finally get
$$ \beta_b(C(qa;1,a))= \alpha (C(qa;1,a))=\dfrac{(q-1)a}{2}.$$

\item $q > a+1$.\\
Note first that if $a+1$ divides $q$, say $q=\ell(a+1)$ for some integer $\ell\ge 1$,
which gives $qa=\ell a(a+1)$,
the result directly follows from Theorem~\ref{th:C((a+1)k;1,a)}
for $k=\ell a$, since
$$\dfrac{ak}{2} = \dfrac{\ell a^2}{2} = \dfrac{qa^2}{2(a+1)} =
\left\lfloor \dfrac{qa^2}{2(a+1)} \right\rfloor.$$

Suppose now that this is not the case, so that $k(a+1) < q < (k+2)(a+1)$, for some odd integer $k\ge 1$.
Let $q= k(a+1) + 2\ell$ (recall that $k$ and $q$ are odd) for some integer $\ell$, $1\le\ell\le a$.
From Proposition~\ref{prop:bound2} we get
\begin{align*}
\beta_b(C(qa;1,a)) \leq \left\lfloor\frac{qa^2}{2(a+1)} \right\rfloor
  & = \left\lfloor\frac{q(a-1)(a+1)+q}{2(a+1)}\right\rfloor
  \\[2ex]
  & = \left\lfloor\frac{q(a-1)(a+1)+(a+1)}{2(a+1)} +\frac{q-(a+1)}{2(a+1)}\right\rfloor
  \\[2ex]
  & = \left\lfloor\frac{q(a-1)(a+1)+(a+1)}{2(a+1)}
     +\frac{(k-1)(a+1) + 2\ell}{2(a+1)}\right\rfloor
  \\[2ex]
  & = \dfrac{(a-1)q+1}{2} + \dfrac{k-1}{2}.
\end{align*}

Moreover, we have $qa=(ka+\ell)(a+1)+\ell(a-1)$ and, thanks to Proposition~\ref{prop:bound4}, we get
$$\beta_b(C(qa;1,a)) \ge (ka+\ell)\frac{a}{2} + \ell(\dfrac{a}{2}-1) =
\frac{ka^2}{2} + \ell(a-1) = \frac{ka^2}{2} + \dfrac{q-k(a+1)}{2}(a-1),$$
which gives
$$\beta_b(C(qa;1,a)) \ge \dfrac{(a-1)q+1}{2} + \dfrac{k-1}{2}.$$

Therefore, by Observation~\ref{obs:beta=alpha}, we finally get
$$\beta_b(C(qa;1,a)) = \alpha(C(qa;1,a)) = \frac{(a-1)q+1}{2} + \dfrac{k-1}{2}= \left\lfloor\frac{qa^2}{2(a+1)} \right\rfloor.$$
\end{enumerate}
\end{enumerate}
This completes the proof.
\end{proof}

Our last general result is the following.

\begin{theorem}\label{th:C(qa+r;1,a)}
 Let $n$, $a$, $q$ and $r$ be integers such that $n=qa+r$, $a$ is even,  $a\ge 6$ and $q\geq \max\{2,r\}$. If either $q$ and $r$ have the same parity, or $q$ and $r$ have different parity and $q+r \geq a-1$, then we have
$$\beta_b(C(n;1,a))  = \alpha(C(n;1,a)) = \left\lfloor\frac{a}{2(a+1)}n \right\rfloor $$
\end{theorem}

\begin{proof}
From Proposition~\ref{prop:bound2}, we get that
$$\sigma(f) \leq \left\lfloor\frac{a}{2(a+1)}\left(n-\frac{a-4}{a}\left|V_f^{2}\right|\right) \right\rfloor \le \left\lfloor\frac{a}{2(a+1)}n \right\rfloor$$
for every $2$-bounded independent broadcast $f$ on  $C(n;1,a)$.
We consider two cases.

\begin{enumerate}

\item  $q $ and $r$ have the same parity.\\
Note first that if $a+1$ divides $q-r$, say $q-r=\ell(a+1)$ for some integer $\ell$,
which gives $qa+r=(q-\ell)(a+1)$,
the result directly follows from Theorem~\ref{th:C((a+1)k;1,a)}
for $k=q-\ell$, since
$$\dfrac{ak}{2} = \dfrac{a(q-\ell) }{2} = \dfrac{(qa+r)a}{2(a+1)} =
\left\lfloor \dfrac{(qa+r)a}{2(a+1)} \right\rfloor.$$

Suppose now that this is not the case, so that $k(a+1) < q- r < (k+2)(a+1)$, for some even integer $k$.
Let $q - r= k(a+1) + 2\ell$ (recall that $k$ and $q -r$ are even) for some integer $\ell$, $1\leq \ell \leq a$. We have
$$\left\lfloor\frac{a}{2(a+1)}(qa+r) \right\rfloor  = \left\lfloor\frac{qa(a+1)-(q-r)a}{2(a+1)} \right\rfloor = \dfrac{aq}{2}
+ \left\lfloor -\frac{(q-r)a}{2(a+1)} \right\rfloor.$$

 Since $-\ell < - \frac{\ell a}{(a+1)} < -\ell +1  $, we get
$$\left\lfloor - \frac{(q-r)a}{2(a+1)} \right\rfloor = \left\lfloor - \frac{(k(a+1) + 2\ell)a}{2(a+1)} \right\rfloor = -\dfrac{ka}{2} + \left \lfloor - \frac{\ell a}{(a+1)} \right\rfloor = - \dfrac{ka}{2} - \ell,$$
which gives
$$\beta_b(C(qa+r;1,a)) \leq \left\lfloor\frac{a}{2(a+1)}(qa+r) \right\rfloor  = \dfrac{aq}{2} - \dfrac{ak}{2} - \ell= \dfrac{(q-k)a}{2}  - \ell. $$

Furthermore, since
$$qa + r=\dfrac{q+r}{2}(a+1)+ \dfrac{q-r}{2}(a-1)=\dfrac{q+r+k(a-1)}{2}(a+1)+\ell(a-1),$$
and thanks to Proposition~\ref{prop:bound4}, we get
$$\beta_b(C(qa+r;1,a)) \ge
\dfrac{q+r+k(a-1)}{2} \left( \frac {a}{2} \right) + \ell (\frac{a}{2} -1 ) = \dfrac{q+r+k(a-1)+ 2\ell }{2} \left( \frac {a}{2} \right)-\ell .$$

Again, since $q-r = k(a+1) + 2\ell $, we have
$$\beta_b(C(qa;1,a)) \ge \dfrac{q+r+q-r - 2k}{2} \left( \frac {a}{2} \right)-\ell = \dfrac{(q-k)a}{2}  - \ell,$$
and thus, thanks to Observation~\ref{obs:beta=alpha}, we finally get
$$\beta_b(C(qa+r;1,a)) = \alpha(C(qa+r;1,a)) = \frac{(a-1)q+1}{2} + \dfrac{k-1}{2}
= \left\lfloor\frac{qa^2}{2(a+1)} \right\rfloor.$$

\item   $q $ and $r$ have different  parity.\\
Similarly to the previous case,  if  $a+1$ divides $q-r$, then
$$\dfrac{ak}{2} =
\left\lfloor \dfrac{(qa+r)a}{2(a+1)} \right\rfloor.$$

Suppose now that this is not the case. We consider two subcases, depending on whether $q-r$ is greater than $a+1$ or not.

\begin{enumerate}
\item $q - r < a+1$.\\
In that case, we have
\begin{align*}
\beta_b(C(qa+r;1,a)) \le \left\lfloor\frac{a}{2(a+1)}(qa+r) \right\rfloor  &= \left\lfloor\frac{qa(a+1)-(q-r)a}{2(a+1)} \right\rfloor\\
 &= \dfrac{aq}{2} + \left\lfloor -\frac{(q-r)a}{2(a+1)} \right\rfloor.
\end{align*}

Since $q - r < a+1$, we get
$$ -\frac{q-r+1}{2} < -\frac{(q-r)a}{2(a+1)}  <  -\frac{q-r+1}{2} + 1,$$
and thus
$$  \beta_b(C(qa+r;1,a)) \le \dfrac{aq}{2} -\frac{q-r+1}{2}. $$

Since $qa + r=\dfrac{q+r+1-a}{2}(a+1)+ \dfrac{q-r +1 +a}{2}(a-1)$, and thanks to Proposition~\ref{prop:bound4}, we have
$$\beta_b(C(qa+r;1,a)) \ge
\left( \dfrac{q+r+1-a}{2} \right) \frac {a}{2}  + \left( \dfrac{q-r +1 +a}{2} \right)  \left(  \frac{a}{2} -1 \right),$$
which gives
$$\beta_b(C(qa;1,a)) \geq \dfrac{aq}{2} -\frac{q-r+1}{2}.$$

Finally, thanks to Observation~\ref{obs:beta=alpha}, we get
$$\beta_b(C(qa;1,a)) = \alpha(C(qa;1,a)) =\dfrac{aq}{2} -\frac{q-r+1}{2}.$$

\item $q - r > a+1$.\\
In that case, we have $k(a+1) < q - r < (k+2)(a+1)$, for some odd integer $k \geq 1 $. Let $q - r= k(a+1) + 2\ell$ (recall that $k$ and $q -r$ are odd) for some integer $\ell$, $1\le\ell\le a$.
Since
$$\left\lfloor - \frac{(q-r)a}{2(a+1)} \right\rfloor = \left\lfloor - \frac{(k(a+1) + 2\ell)a}{2(a+1)} \right\rfloor = -\dfrac{ka}{2} + \left \lfloor - \frac{\ell a}{(a+1)} \right\rfloor $$
and   $-\ell < - \frac{\ell a}{(a+1)} < -\ell +1  ,$
we get
$$ \left\lfloor\frac{a(qa+r)}{2(a+1)} \right\rfloor  = \left\lfloor\frac{qa(a+1)-(q-r)a}{2(a+1)} \right\rfloor = \dfrac{aq}{2}
+ \left\lfloor -\frac{(q-r)a}{2(a+1)} \right\rfloor= \dfrac{aq}{2} - \dfrac{ka}{2} - \ell.$$


This implies
$$\beta_b(C(qa+r;1,a)) \le \left\lfloor\frac{a(qa+r)}{2(a+1)} \right\rfloor  = \dfrac{aq}{2} - \dfrac{ak}{2} - \ell= \dfrac{(q-k)a}{2}  - \ell .$$

Moreover, since $qa + r=\dfrac{q+r+1-a}{2}(a+1)+ \dfrac{q-r+a+1}{2}(a-1)$ and  $q - r= k(a+1) + 2\ell$, we get
$$ qa + r = \dfrac{q+r+k(a-1)}{2}(a+1)+ \ell(a-1). $$

Hence, we have
$$\beta_b(C(qa+r;1,a)) \ge
\dfrac{q+r+k(a-1)}{2} \cdot \frac{a}{2} + \ell (\frac{a}{2} -1 ) = \dfrac{q+r+k(a-1)+ 2\ell }{2} \cdot \frac{a}{2} -\ell,$$
which gives
$$\beta_b(C(qa+r;1,a)) \ge \dfrac{q+r+q-r -2k  }{2} \cdot \frac {a}{2} -\ell = \dfrac{(q-k)a}{2}  - \ell.$$

Finally, thanks to Observation~\ref{obs:beta=alpha}, we get
$$\beta_b(C(qa+r;1,a)) = \alpha(C(qa+r;1,a)) =
 \left\lfloor\frac{qa^2}{2(a+1)} \right\rfloor.$$
\end{enumerate}

\end{enumerate}
This completes the proof.
\end{proof}

\section{Discussion}\label{sec:discussion}

We proved that every circulant graph of the form $C(n;1,a)$,
$3\le a\le \lfloor\frac{n}{2} \rfloor$, admits a $2$-bounded $\beta_b$-broadcast,except when $n=2a+1$, or $n=2a$ and $a$ is even.
Using this property, we determined the exact value of the broadcast independence number of several classes of circulant graphs of the form $\beta_b(C(n;1,a))$, $2\le a\le\lfloor\frac{n}{2}\rfloor$.
In several cases, we showed that $\beta_b(C (n;1,a))$ reaches one of its lower bounds, namely $\alpha(C(n;1,a))$ or $2(\diam(C(n;1,a))-1)$. In particular, whenever $\beta_b(C (n;1,a))=\alpha(C(n;1,a))$, we get that $C (n;1,a)$ admits a $1$-bounded $\beta_b$-broadcast.

We finally mention a few open problems that seem worth to be investigated.
\begin{enumerate}
\item Determine the value of $\beta_b(C (n;1,a))$ for the remaining unsolved cases namely:
\begin{enumerate}
    \item For odd integers $a$ and $n$, with $a \geq 5$ and $a\not|\ n$.
    \item For integers $n$, $a$, $q$ and $r$, with $n=qa+r$, $a\geq 6$ is even, $n$ is  divisible neither by $a$ nor by $a+1$, and either
           \begin{enumerate}
           \item $q<r$, or
           \item $q > r$, $q$ and $r$ have different parity, $q+r < a-1$.
           \end{enumerate}
\end{enumerate}

\item Determine the broadcast independent number of other classes of circulant graphs.
\item Characterize the classes of graphs $G$ for which $\beta_b(G) = \alpha(G)$ or $\beta_b(G) = 2(\diam(G)-1)$, respectively.
\end{enumerate}

\noindent
{\bf Acknowledgment.} The first and second authors acknowledge General Directorate of Scientific Research and Technological Development of the Algerian Ministry of Higher Education and Scientific Research (DGRSDT) for support of this work.

\end{document}